\documentclass[12pt, a4paper, ngerman, UKenglish]{scrartcl}

\usepackage[utf8]{inputenc}
\usepackage[ngerman, main=UKenglish]{babel}
\usepackage[pdftex,hyperfootnotes=false,pdfpagelabels, hidelinks,psdextra]{hyperref}
\usepackage[dvipsnames]{xcolor}
\usepackage{graphicx}

\usepackage{enumerate}
\usepackage{csquotes}
\usepackage{comment}
\usepackage{xspace}

\usepackage{setspace}
\usepackage{relsize}

\newlength{\alphabet}
\usepackage{geometry}
\usepackage[backend=biber, maxbibnames=5]{biblatex}

\usepackage{amsfonts}
\usepackage{libertinus}
\usepackage{libertinust1math}
\usepackage[T1]{fontenc}
\usepackage{eucal}
\usepackage{beramono}
\usepackage[semibold]{sourcesanspro}
\usepackage{sectsty}
    \allsectionsfont{\sffamily}
\usepackage{bm}

\usepackage{textcase}
\newcommand{\flatcaps}[1]{\large{\textsc{\MakeTextLowercase{#1}}}}

\usepackage[subfigure]{tocloft}

\usepackage{ragged2e}
\usepackage[protrusion=true, expansion=true]{microtype}

\nobreak

\PassOptionsToPackage{norule, ragged, hang}{footmisc}
\RequirePackage{fnpct}
\RequirePackage{footmisc}

\makeatletter
\renewcommand\@makefntext[1]{
    \@thefnmark.~#1}
\makeatother

\usepackage{booktabs}
\usepackage{multirow}
\usepackage{longtable}
\usepackage{rotating}

\usepackage{wrapfig}
\usepackage{graphicx}
\usepackage{grffile}
\usepackage[font=small, labelfont={sf, bf}]{caption}
\usepackage{subcaption}

\usepackage{mathtools}
    \mathtoolsset{showonlyrefs}
\usepackage{amsthm}
\usepackage{amsmath}
\usepackage{amssymb}
\usepackage{tensor}
\usepackage{tikz-cd}
\usepackage{mathdots}
\usepackage{nicefrac}
\usepackage{xfrac}
\usepackage[lined, algoruled, english]{algorithm2e}
    \SetAlgoCaptionSeparator{.}
    
\numberwithin{equation}{section}

\usepackage{lipsum}

\newtheoremstyle{TEMPLATE}
{} 
{} 
{} 
{} 
{} 
{} 
{} 
{} 

\newtheoremstyle{majorStatement}
    {\topsep}
    {\topsep}
    {\itshape}
    {}
    {\bfseries}
    {}
    {\newline} 
    {\sffamily\bfseries\thmname{#1} \thmnumber{#2}
     \nopagebreak \medskip \nopagebreak}

\theoremstyle{majorStatement}
\newtheorem{theorem}{Theorem}[section]
\newtheorem{proposition}[theorem]{Proposition}

\newtheoremstyle{minorStatement}
    {\topsep}
    {\topsep}
    {\normalfont}
    {0pt}
    {\bfseries}
    {.~}
    {0pt} 
    {\sffamily\thmname{#1} \thmnumber{#2}\thmnote{\normalfont (#3)}}
    
\theoremstyle{minorStatement}
\newtheorem{lemma}[theorem]{Lemma}
\newtheorem{corollary}[theorem]{Corollary}
\newtheorem{remark}[theorem]{Remark}

\newtheoremstyle{titledDefinition}
    {\topsep}
    {\topsep}
    {\normalfont}
    {0pt}
    {\bfseries}
    {}
    {\newline} 
    {\smallskip\sffamily\thmname{#1} \thmnumber{#2}\thmnote{\normalfont\itshape ~--~#3~--}}

\theoremstyle{titledDefinition}
\newtheorem{definition}[theorem]{Definition}

\newtheoremstyle{titledTheorem}
    {2\topsep}
    {2\topsep}
    {\itshape}
    {}
    {\bfseries}
    {}
    {\newline} 
    {\thmnote{\normalfont\flatcaps{#3}}
     \nopagebreak \medskip \nopagebreak}

\theoremstyle{titledTheorem}

\newtheoremstyle{assumptionStyle}
    {\topsep}
    {\topsep}
    {\normalfont}
    {0pt}
    {\bfseries}
    {.~}
    {0pt}
    {\sffamily\thmname{#1} \thmnumber{#2}}

\theoremstyle{assumptionStyle}
\newtheorem{assumption}[theorem]{Assumption}

\renewenvironment{proof}
    {\par\vspace{\topsep}{\sffamily\bfseries Proof.~}}{\nopagebreak\par\hfill\qed}

\newcommand{\coloneqq}{:\mspace{-1.5mu}=}
\newcommand{\eqqcolon}{=\mspace{-1.5mu}:}

\newcommand{\Matern}{Ma{\-}t{\'e}rn\xspace}
\newcommand{\KL}{Kar{\-}hu{\-}nen-Lo{\'e}ve\xspace}

\newcommand{\R}{\mathbb{R}}
\newcommand{\rpow}[1]{ \mathbb{R}^{ #1 } }
\newcommand{\C}{\mathbb{C}}

\newcommand{\N}{\mathbb{N}}

\newcommand{\Z}{\mathbb{Z}}

\newcommand{\sumlim}[2]{\sum\limits_{#1}^{#2}}
\newcommand{\suml}[1]{\sum\limits_{#1}}
\newcommand{\prodlim}[2]{\prod\limits_{ #1 }^{ #2 } \mspace{-1mu}}

\newcommand{\seq}[2]{( #1 )_{ #2 \in \mathbb{N}}}

\newcommand{\conj}[1]{\overline{#1}}

\newcommand{\mc}[1]{\mathcal{ #1 }}
\newcommand{\mf}[1]{\mathfrak{ #1 }}
\newcommand{\mbb}[1]{\mathbb{ #1 }}
\newcommand{\mrm}[1]{\mathrm{ #1 }}

\newcommand{\bidx}[2]{#1^{({#2})}}

\newcommand{\pfwd}[2]{{ #1 }_{\mathlarger{\sharp}}\,{ #2 }}

\newcommand{\iprod}[1]{\mspace{2mu} \langle {#1} \rangle}

\newcommand{\ltwo}[1]{L^{2}( #1 )}

\newcommand{\adjnt}[1]{{#1}^{\dagger}}

\newcommand{\norm}[1]{\| \mspace{2mu} {#1} \mspace{2mu}\|}

\newcommand{\fourier}[1]{\widehat{#1}}
\newcommand{\farg}[2]{\widehat{#1} \mspace{-1mu} \left({#2}\right)}

\newcommand{\prdi}[2]{{#1}_{#2}^{(\pi)} \mspace{-1mu}}

\DeclareMathOperator{\exv}{\mathbb{E} \mspace{-1mu}}

\newcommand{\stdnrm}[1]{\mathcal{N}\big({#1}\big)}

\usepackage{accents}

\newcommand{\pcoeff}[1]{\farg{\varphi}{(2 \alpha)^{-1} {#1}}}
\newcommand{\nrf}[2]{\bidx{ {#1}_{#2} }{0}\mspace{-2mu}}

\newcommand{\drf}[2]{\bidx{ {#1}_{#2} }{1}\mspace{-2mu}}

\newcommand{\e}[2]{e^{({#1})}_{#2}}

\title{Dirichlet-Neumann Averaging: The DNA of Efficient Gaussian Process Simulation}
\author{Robert~Kutri\thanks{Institute for Mathematics and Interdisciplinary Center for Scientific
Computing (IWR), Heidelberg University, 69120, Heidelberg, Germany~(\href{mailto:robert.kutri@uni-heidelberg.de}
{robert.kutri@uni-heidelberg.de},~\href{mailto:robert.scheichl@uni-heidelberg.de}{r.scheichl@uni-heidelberg.de})} \and
Robert~Scheichl\footnotemark[1]}

\date{\today\\[1em]
\raggedright {\small
    \textbf{\textsf{MSC Codes}}:~60G15,~60G60,~60G10,~60-08,~65C20 \\[1em]
    \textbf{\textsf{Funding:}} This work is supported by the Deutsche Forschungsgemeinschaft (German Research\\
    Foundation) under Germany's Excellence Strategy EXC 2181/1 - 390900948 (the Heidelberg \\[-0.8em]
    STRUCTURES Excellence Cluster)
    }}

\begin{document}

\maketitle

\vspace{-2em}

\begin{abstract}
\small
    Gaussian processes (GPs) and Gaussian random fields (GRFs) are
essential for modelling spatially varying stochastic phenomena.
Yet, the efficient generation of corresponding realisations
on high-resolution grids remains challenging,
particularly when a large number of realisations are required.
This paper presents two novel contributions.
First, we propose a new methodology based on Dirichlet-Neumann
averaging (DNA) to generate GPs and
GRFs with isotropic covariance on regularly spaced grids. 
The combination of discrete
cosine and sine transforms in the DNA sampling approach allows
for rapid evaluations
without the need for modification or padding of the desired
covariance function. While this introduces an error in the
covariance, our numerical experiments show that this error is
negligible for most relevant applications, representing a
trade-off between efficiency and precision. We provide
explicit error estimates for \Matern covariances.
The second contribution links our new methodology to the stochastic
partial differential equation (SPDE) approach for sampling GRFs.
We demonstrate that the concepts developed in our methodology can also guide
the selection of boundary conditions in the SPDE framework. We
prove that averaging specific GRFs sampled via the SPDE approach yields
genuinely isotropic realisations without domain extension, with the
error bounds established in the first part remaining valid.

\end{abstract}

\vspace*{\fill}

\pagebreak

 The generation of realisations of Gaussian Processes (GPs) or Gaussian Random
Fields (GRFs) with a prescribed covariance is a crucial step in many
algorithms, 
e.g., related to
spatial statistical inference
\cite{bhatt2015effect, griffiths2009influence, minden2017fast, simpson2016going, VILLEJO2023100744},
Bayesian non-linear inverse problems
\cite{cotter2013pcn, nickl2023bayesian},
classification
\cite{williams1998bayesian}
and uncertainty quantification
\cite{cliffe2011multilevel, roininen2014whittle}.
In these contexts, GPs and GRFs can both be viewed as collections of
Gaussian random variables $u_x$, indexed by a continuous variable
$x \in D \subseteq \R^d$ with $d \in \N$. In what follows, we
primarily adopt the perspective associated with GRFs, interpreting
$u$ as a random function over the domain of interest $D$.
Moreover, we consider the case of GRFs with isotropic covariance function.
In practical applications, one is typically interested in generating realisations
of a given GRF on a finite-dimensional grid $\{x_j : j = 1, \ldots, n\} \subset D$,
for $n \in \N$. By definition, any marginal distribution of a GRF, when evaluated
at a finite set of $n$ points, is an $n$-dimensional multivariate Gaussian with
moments determined by the underlying GRF. 
Thus, we can re-frame the GRF sampling problem on a grid, as seeking to generate 
realisations $z \in \R^n$ of a multivariate Gaussian random variable with 
covariance matrix $C \in \R^{n \times n}$, whose structure is inherited from the GRF.
Generating realisations of such $z \in \R^n$, even
for isotropic covariance, is computationally challenging when $n$
is too large for a direct factorisation of the covariance matrix $C$.
Although efficient sampling algorithms exist for certain settings, they often become
impractical when both small error margins and good scalability on parallel architectures
are required.

\KL expansions \cite{schwab2006karhunen, khoromskij2009application} offer a convenient
stochastic parametrisation in terms of the eigenfunctions of the covariance operator.
However, the number of modes required for a given accuracy depends on the decay of the eigenvalues of the
covariance operator. Computation and storage of these modes can become prohibitively costly for large GRFs
or for fields with low regularity. Similar limitations appear in reduced-basis or low-rank sampling
approaches \cite{latz2019fast, feischl2018fast, HARBRECHT2012428}.

For isotropic covariance functions, the matrix $C \in \rpow{n \times n}$ is (block-) Toeplitz
and can thus be embedded in a $M \times M$ matrix with $M > n$,
which is (block-) circulant. Circulant matrices in turn, are diagonalised by
a Fast Fourier Transform (FFT). If the FFT frequencies corresponding to the auxiliary
matrix are non-negative, this leads to an efficient algorithm for generating  GRF
realisations.
This is the
underlying idea in Circulant Embedding (CE) methods \cite{wood1994simulation, dietrich1997fast}.
The primary obstacle to computational efficiency for this
approach stems from the fact that in order to achieve non-negativity of the FFT frequencies,
the size $M$ of the auxiliary matrix may need to be substantially larger than $n$.
This results in significant computational and memory overhead as the embedding
size requires fine-tuning via padding.
Sufficient conditions for non-negativity of the FFT frequencies have
been established in \cite{wood1995truncated, gneiting1998simple}. The authors
in \cite{graham2018analysis} derive explicit lower bounds on the size of $M$ for 
\Matern covariance functions.

By modifying the covariance function outside the domain of interest, it may be possible to alter the
spectral density in a way that achieves non-negativity of the FFT frequencies with a smaller $M$ compared
to CE. This idea has been pursued e.g.~in \cite{gneiting2006fast}
and \cite{moreva2018fast}. The authors in \cite{bachmayr2018representations}  have placed this idea
in a framework based on periodic continuations of the covariance function.
A rigorous analysis and explicit lower bounds on $M$ for classical CE and 
a smooth periodisation of \Matern covariance functions is carried out in \cite{bachmayr2020unified}.
All of these periodic continuations alter the covariance function in such a way as to achieve compact support.

The first main result of this paper in Section~\ref{sec:sampFr} is a novel methodology for generating GRF realisations
based on a periodisation of the covariance function, that does not rely on compact support and is positive-definite
for any choice of $n$. Evaluation of the corresponding GRF is performed by a combination of 
Discrete Sine and Cosine Transforms. This Dirichlet-Neumann Averaging (DNA) approach provides a fast and
well-scaling method for the generation of GRF realisations with arbitrary isotropic covariance function. In 
Proposition~\ref{prop:combRFCov} we show that the covariance of this GRF is indeed a specific periodisation
of the desired covariance with more favourable characteristics compared to a naive periodisation.
Thus, the DNA methodology facilitates fast evaluations of the GRF, but its practicability depends
on the error introduced by the periodisation.
In Proposition~\ref{prop:covErrorMatern}, we provide an explicit error bound for the case of \Matern covariances.
This shows exponential decay of the error with respect to a parameter $\alpha$, which plays the role of $M$
in the CE method. In comparison to CE, this framework represents a trade-off between an additional error
in the covariance and the computational speed-up gained through eliminating the need for padding.
In this regard, numerical experiments in Section~\ref{sec:numExp} suggest that the additional error is
essentially negligible for most practically relevant parameter configurations.

Although this method of evaluation is fast, it requires a regularly spaced grid, and the use of the FFT may
limit its scalability in massively parallel settings.
For these applications, we establish a link to the stochastic partial differential equation (SPDE) approach for
sampling GRFs with \Matern covariance introduced in \cite{lindgren2011explicit}. In that approach,
approximations to GRFs with \Matern covariances are computed by solving a fractional elliptic SPDE
on the domain $D$. The primary advantage of the SPDE approach 
is that it leverages
the flexible, fast, and scalable Finite-Element (FE) framework for solving PDEs, enabling efficient and parallel
generation of GRFs.

However, aliasing effects introduced by the artificial choice of boundary conditions
for the SPDE pose a significant obstacle \cite{lindgren2011explicit, roininen2014whittle}. A common remedy to this is to
first compute the solution on a larger domain $\widetilde{D} \supset D$ and then truncate back to the domain of
interest $D$
\cite{khristenko2019analysis, croci2018efficient, drzisga2017scheduling}.
This is similar in spirit to padding in CE.
Addressing this error can
introduce significant computational overhead, which may offset the benefits of the SPDE framework. The authors in
\cite{khristenko2019analysis} analyse the maximal covariance error in the case of homogeneous Dirichlet,
homogeneous Neumann, as well as Robin boundary conditions on the entire boundary.
They derive upper bounds on the covariance error for a given extension of the desired domain, in terms 
of values of the \Matern covariance function itself. They demonstrate that none of the aforementioned boundary
conditions, when considered in isolation, result in isotropic fields.

Thus, the paper's second main result in Proposition~\ref{prop:SPDECorrespondence} establishes a
correspondence between GRFs in the SPDE approach and the periodisation established in
Section~\ref{sec:frame:spdeConnection}. This correspondence can be leveraged to eliminate the need for
oversampling and to generate truly isotropic random fields. To achieve this, instead of  oversampling the
SPDE solution on a larger domain, in our DNA variant of the SPDE approach, several SPDE solutions
with Dirichlet and Neumann 
boundary conditions are averaged on the original, 
smaller domain.  As a consequence, in massively parallel environments, where memory
constraints are often the bottleneck and where it is essential to reuse existing data structures and
communication patterns, the proposed new  procedure has clearly favourable characteristics compared to oversampling.
Moreover, the explicit error bound of Proposition~\ref{prop:covErrorMatern}
extends also to the SPDE version of the DNA approach. 
Numerical experiments in Section~3 indicate that this periodisation error is negligible
compared to the discretisation error from the FE method across a large range of mesh resolutions
of the underlying mesh.

The paper is structured as follows.
In Section~1, we introduce notation, the
problem setting, and the periodisation framework. We finish with
a brief introduction to the SPDE approach. In Section~\ref{sec:sampFr}, we construct the new DNA sampling
framework and state the main results of the paper. Section~\ref{sec:higherDim} concerns the general 
construction of the new sampling approach and the error estimates for \Matern covariances, while its
connection to the SPDE approach is described in Section~\ref{sec:frame:spdeConnection}.
Section~\ref{sec:numExp} presents numerical experiments that illustrate the developed theory.

\section{Preliminaries}

    We use the following convention for the $d$-dimensional Fourier transform $\fourier{\varphi}: \R^d \to \C$ of a
function $\varphi \in L^1(\R^d)$:
$$\fourier{\varphi}(y) \coloneqq \int_{\R^d} \varphi(x)~e^{-2 \pi i y \cdot x} \, \mrm{d}x.$$
where $x \cdot y$ denotes the Euclidean inner product for any two vectors $x, y \in \R^d$.
If $\varphi$ is symmetric about the origin and invariant with respect
to rotations, we can represent it in terms of a function $\rho: \R \to \R$ via
$\varphi(x) = \rho(\norm{x}_2)$ for any $x \in \R^d$.
In this case, the Fourier transform $\fourier{\varphi}$ can be expressed in terms of 
the Hankel transform $\mc{H}_d \rho: \R \to \R$ of $\rho$ through
$\fourier{\varphi}(y) = \big(\mc{H}_d \rho\big)(\norm{y}_2)$, where
$$\big(\mc{H}_d \rho\big)(r) \coloneqq 2 \pi r^{-\nicefrac{(d-2)}{2}}
        \int_{\R}\rho(t)~t^{\nicefrac{d}{2}} \, J_{\nicefrac{(d-2)}{2}}(2 \pi r t)~\mrm{d}t, \quad \text{for} \ \ r \geq 0,$$
and $J_{\gamma}$ denotes the Bessel function of order $\gamma \in \R$. Regarding computational
aspects of Fourier analysis, the Fast Fourier Transform (FFT) allows for the efficient evaluation of
the Discrete Fourier transform.
Similarly, one can use the (type 1) Discrete Cosine Transform (DCT) and the 
(type 1) Discrete Sine Transform (DST) to efficiently evaluate the sums
\begin{equation}
    \sumlim{k=0}{n-1} \eta_k \cos\Big(\frac{\pi m k}{n-1}\Big),
        \quad \text{and} \quad
        \sumlim{k=1}{n-1} \eta_k \sin\Big(\frac{\pi m k}{n-1}\Big),
        \quad \text{for } \eta_k \in \R, \quad m = 0, \ldots, n-1.
\end{equation}
The basis functions used in the FFT, DCT and DST satisfy periodic,
homogeneous Neumann and homoegeneous Dirichlet boundary conditions respectively.

Corresponding higher-dimensional transforms are constructed through the tensor product of
one-dimensional transforms. Multiindices $\mu$, which may be understood as vectors
in the $d$-dimensional lattice $\Z^d$, form an essential component of the notation used in this context.
For multiindices $\mu = (\mu_1, \ldots, \mu_d)^T \in \Z^d$, we will mainly use the norm
$\norm{\mu}_{\infty} \! \coloneqq \! \max\{|\mu_j| : j = 1, \ldots, d\}$ and the 
Euclidean norm $\norm{\mu}_2$. For fixed $n \in \N$,
we also define the truncated lattices
$\Z_n^d \coloneqq \{\mu \in \Z^d : \norm{\mu}_{\infty} \leq n \}$,
and
$\N_n^d \coloneqq \{\mu \in \N^d : \norm{\mu}_{\infty} \leq n\}$.
Similar to decomposing $\Z^2$ into disjoint quadrants, we can decompose
$\Z^d$ into $2^d$ hyperoctants $H_{q}$, where for a given $q \in \{-1, 1\}^d$, we define
\begin{equation}    \label{eq:hyperoctant}
    H_{q} \coloneqq \{ \mu \in \Z^d : \mu_j \geq 0 \text{ when } q_j = 1 \text{ and }
        \mu_j < 0, \text{ when } q_j = -1 \}.
\end{equation}
The truncated hyperoctants $H_{q, n}$ for $n \in \N$ are defined analogously to
$H_q$ in \eqref{eq:hyperoctant}, where we replace $\Z^d$ by $\Z_n^d$. It follows from
their construction that
\begin{equation}    \label{eq:hyperoctantDecomposition}
    \Z^d = \bigcup\limits_{q \in \{-1, 1\}^d} H_{q} \quad \text{and} \quad
        \Z_n^d = \bigcup\limits_{q \in \{-1, 1\}^d} H_{q, n}, 
\end{equation}
where the unions are disjoint.

We denote the space of continuous functions on $U$, where $U \subset \R^d$
is bounded and connected, by $C(U)$. Further, we denote the space of
equivalence classes of square-integrable functions that agree almost everywhere w.r.t.~the 
Lebesgue measure on $\R^d$ by $L^2(U)$. In the usual abuse of notation, we treat elements $f \in L^2(U)$ 
as functions in the classical sense. The space $L^2(U)$ is a separable Hilbert space and
we denote the associated inner product by $\iprod{\cdot, \cdot}$ and the induced
norm by $\norm{\cdot}_{L^2(U)}$. The Sobolev space $H^1(D)$ of functions $f \in L^2(U)$
whose weak partial derivatives $\partial_j f$, $j = 1, \ldots, d$ exist and lie in $L^2(U)$
is denoted by $H^1(U)$ (see e.g.~\cite{evans2010_pde}, Chapter~5 for details).

    \subsection{Problem Setting}    \label{sec:prel:probSetting}
        We begin by providing a more precise definition of the concepts of isotropic GRFs and the
associated covariance functions, which were outlined in the introductory section.
For the given domain $D \subset \R^d$, a GRF may be understood as a 
random variable taking values in some function space on $D$, say 
$\ltwo{D}$. That is, a GRF is a measurable map
$\mc{U}$ from a latent probability space $(\Omega, \mf{A}, \mbb{P})$ into $(\ltwo{D}, \mf{B})$, 
where $\mf{B}$ denotes a suitable $\sigma$-algebra on $\ltwo{D}$.
Regarding the technical details of this viewpoint, we refer to the books by Gelfand and 
Vilenkin \cite{gelfand1964_generalizedFunctions4} or Adler
\cite{adler2010_geometryOfRandomFields} and will from here on assume that such a 
construction is possible in our setting. As we are interested in point-wise evaluations of GRF
realisations, we restrict ourselves to GRFs $\mc{U}$, for which the push-forward measure
$\pfwd{\mbb{P}}{\mspace{-1mu}\mc{U}}$ is supported on $C(D)$
(c.f.~Theorem 2.2.1 in \cite{adler2010_geometryOfRandomFields}). In this case, evaluating
a GRF at any $x \in D$ is well-defined and we can represent a GRF by the function
\begin{equation}
    u: D \times \Omega \to \R, \quad u(x, \omega) \coloneqq \big[\mc{U}(\omega)\big](x),
        \quad \omega \in \Omega, x \in D.
\end{equation}
We will not explicitly state $\omega$ when it is convenient to do so
and causes no confusion. 

Analogously to finite-dimensional, multivariate Gaussian distributions, a GRF is fully characterised by 
its mean function $m: D \to \R$ and its covariance function $\mrm{cov}: D \times D \to \R$, which
are defined through
\begin{equation}
    m(x) \coloneqq \exv u(x, \cdot), \quad \text{and} \quad
        \mrm{cov}(x, y) \coloneqq \exv u(x, \cdot) u(y, \cdot).
\end{equation}
As mentioned previously, we restrict ourselves to centred GRFs, i.e., GRFs where
$m \equiv 0$.
Furthermore, we focus on \emph{isotropic} covariance functions, which encode translational, as well as
rotational invariance of the covariance function.
\begin{definition}[Isotropic Covariance Function]   \label{def:isoCov}
    We call the covariance function of a GRF \emph{isotropic}, if there
    exists an $\ell > 0$ and $\rho: \R \to \R$, such that
    \begin{equation}    \label{eq:isoCovDef}
        \mrm{cov}(x, y) = \rho\big(\ell^{-1} \norm{x - y}_2 \big), \quad \text{for all } x, y \in D.
    \end{equation}
    We refer to $\ell$ as the \emph{correlation length} associated to the covariance function.
\end{definition}
For an isotropic covariance \eqref{eq:isoCovDef}
the Fourier transform of the function
\begin{equation}    \label{eq:statCovDef}
    \varphi: \rpow{d} \to \R, \quad \text{with}\quad \varphi(\delta) \coloneqq \rho(\ell^{-1} \norm{\delta}_2),
\end{equation}
can be expressed via the Hankel transform $\mc{H}_{d}$ of $\rho$
such that
$\fourier{\varphi}(\zeta) = \ell^{\,d} \big[\mc{H}_d\rho\big]\big(\ell~\norm{\zeta}_2\big)$, $\zeta \in \rpow{d}$.
In order to distinguish the representations of the covariance function via $\varphi$ and $\rho$, we will
refer to $\varphi$ as the \emph{stationary} covariance function and to $\rho$ as the \emph{isotropic}
covariance function.
A widely used family of covariance functions is given by the \Matern family with smoothness
parameter $\nu > 0$, which are defined by
\begin{equation}    \label{eq:maternDefinition}
    \rho(s) = \frac{2^{1 - \nu}}{\Gamma(\nu)}
        \big(\sqrt{2 \nu} s \big)^{\nu} K_{\nu}(\sqrt{2 \nu} s),
        \quad s \in [0, \infty).
\end{equation}
Here, $\Gamma$ denotes the Euler Gamma function and $K_{\nu}$ is the Bessel function of second kind and order
$\nu$. For the \Matern covariance kernel, the corresponding Hankel transform is 
explicitly given by (see e.g.~\cite{graham2018analysis}, Equation (2.22)):
\begin{equation}    \label{eq:maternHankel}
    \big(\mc{H}_d\rho \big)(s) = C_{\nu} (2 \nu)^{\nu} \big(2 \nu + (2 \pi s)^2\big)^{-(\nu + \nicefrac{d}{2})},
        \quad \text{where} \quad
        C_{\nu} \coloneqq (4 \pi)^{\nicefrac{d}{2}} 
            \frac{\Gamma(\nu + \nicefrac{d}{2})}{\Gamma(\nu)}.
\end{equation}

The transition from the infinite-dimensional formulation that has been presented thus far, to a finite-dimensional
representation which is typically required in applications, is achieved by considering evaluations of GRFs on
a given grid $T_N \coloneqq \{x_j \in D: j = 1, \ldots, N\}$, $N \in \N$. To be more precise, if we consider the
random variable $z: \Omega \to \R^N$, whose components are given by
$(z(\omega))_j \coloneqq u(x_j, \omega)$, for $j = 1, \ldots, N$,
then $z$ is a multivariate Gaussian random variable by construction. Given that we assume $u$ is centred, it follows that $\exv z = 0$. In this
case, the structure of the covariance matrix $C \in \R^{N \times N}$ of $z$ is inherited solely from the covariance
function of $u$. For isotropic covariance functions \eqref{eq:isoCovDef} this explicitly reads
\begin{equation}    \label{eq:componentwiseIsotropicCovarianceMatrix}
    (C)_{k\ell}  = \exv (z)_k (z)_\ell = \exv u(x_k, \cdot) \, u(x_\ell, \cdot) = \rho(\ell^{-1} \norm{x_k - x_\ell}_2),
        \quad k, \ell = 1, \ldots, N.
\end{equation}
Thus, in the finite-dimensional setting, we consider the problem of efficiently generating realisations
$z \in \R^N$ of a centred multivariate Gaussian random variable, with covariance matrix given by $C$ as in
\eqref{eq:componentwiseIsotropicCovarianceMatrix}.
One common approach for generating such realisations is to find a factorisation
$C = H H^T$, for some $H \in \R^{N \times N}$. The random vector defined through $z \coloneqq H \xi$, where
$\xi \sim \mc{N}(0, \mbb{I}_{N \times N})$ is a multivariate standard normal random variable,
then satisfies 
\begin{equation}    \label{eq:covFactorisation}
    \exv z = 0 \quad \text{and} \quad \exv z z^T = H \big(\exv \xi \xi^T\big) H^T = C.
\end{equation}
We assume that for our purposes, $N$ is so large that a simple Cholesky factorisation of $C$ is not feasible.
Accordingly, we seek factorisations that scale significantly better than cubic as $N$ tends to infinity
-- ideally linearly with $N$.
To this end, leveraging the interpretation of $z$ as the evaluation of an underlying GRF with 
isotropic covariance $\rho$ is crucial in designing a scalable algorithm. 

In addition to the GRF $u$ being centred and with isotropic covariance function $\rho$, we require assumptions
that ensure the series and expressions occurring later in the text to be well-defined. To this end, we use the
following assumptions, which we frame in terms of the stationary covariance function $\varphi$ as defined in
\eqref{eq:statCovDef}.
\begin{assumption}  \label{assump:statCov}
    We assume that any occurring stationary covariance function $\varphi \in L^1(\R^d) \cap L^2(\R^d)$ satisfies
        $|\fourier{\varphi}(y)| \leq A (1 + \norm{y}_2)^{-(d + \varepsilon)}$
            for $A, \varepsilon > 0$ and all $y \in \R^d$.
\end{assumption}
This somewhat technical regularity assumption is satisfied by most covariance functions used in
practice. For instance, recalling \eqref{eq:statCovDef} and \eqref{eq:maternHankel}, the \Matern covariance
function \eqref{eq:maternDefinition} satisfies Assumption~\ref{assump:statCov}; for details see
Appendix~\ref{app:lem:maternAssumption}.

    \subsection{Periodisation-Based Sampling}   \label{sec:prel:prd}
        The efficiency of CE and related methods is derived from the application
of the Fast Fourier Transform (FFT). 
In using the FFT, one implicitly evaluates a periodic continuation of the covariance function.
By appropriately modifying the covariance function, it is possible to alter specific aspects of
the sampling method. However, essentially all methods that 
build upon this idea focus on achieving compact support of the covariance
function, see e.g.~\cite{gneiting2006fast, moreva2018fast, bachmayr2018representations}.
Here, we take a slightly different approach and use
a natural periodic continuation, which is implied by the Poisson Summation formula (PSF)
(Theorem~VII.2.4 and Corollary in \cite{stein1971introduction}).
This approach has the advantage of providing a periodic continuation that is always positive definite.
However, it also results in a modification of the covariance function inside the domain of interest $D$.
The accuracy of the method will thus hinge on how closely this modified covariance function matches the
desired covariance function in $D$.

The PSF states that if $\varphi: \R^d \to \R$ satisfies Assumption~\ref{assump:statCov}, then
\begin{equation}    \label{eq:poisSum}
    \suml{\eta \in \Z^d} \varphi(x + \eta) = \suml{\mu \in \Z^d} \fourier{\varphi}(\mu)~e^{-2 \pi i \mu \cdot x}
\end{equation}
and both series in \eqref{eq:poisSum} converge absolutely.
Informally, the series on the left of \eqref{eq:poisSum} represents
an infinite sum of copies of the pristine covariance $\varphi$, where
each copy is translated by a vector $\eta$ of the $d$-dimensional
lattice $\Z^d$. The PSF then provides an
explicit Fourier series representation of the resulting periodic function.
We capture this notion of periodisation in the following definition,
which generalises \eqref{eq:poisSum} slightly, by introducing a scaling parameter
$\alpha \geq 1$.

\begin{definition}[Periodisation]   \label{def:periodisation}
    Let $\varphi: \R^d \to \R$ satisfy Assumption~\ref{assump:statCov} and let
    $\alpha \geq 1$ be a given scaling parameter.
    We define the \emph{periodisation} $\prdi{\varphi}{\alpha}$ of $\varphi$
    through the Fourier series
    \begin{equation}    \label{eq:periodisationDefinition}
        \prdi{\varphi}{\alpha}(\delta) \coloneqq
            \alpha^{-d} \suml{\mu \in \Z^d}
            \farg{\varphi}{\alpha^{-1} \mu}
            e^{-2 \pi i \alpha^{-1} \mu \cdot \delta}. 
    \end{equation}
    We refer to $\prdi{\varphi}{\alpha}$ as \emph{positive definite},
    if $\farg{\varphi}{\alpha^{-1} \mu} > 0$, for all $\mu \in \Z^d$.
\end{definition}
It follows from \eqref{eq:poisSum}, that
\begin{equation}    \label{eq:PSFstationaryCovariance}
    \prdi{\varphi}{\alpha}(\delta)
        = \suml{\eta \in \Z^d} \varphi(\delta + \alpha \eta),
        \quad \delta \in \R^d.
\end{equation}
Stationary covariance functions are positive definite
functions on $\R^d$ by definition. Thus, Bochner's theorem \cite{bochner1932Fourier}
ensures that any periodisation $\prdi{\varphi}{\alpha}$ of a stationary covariance function
is also positive definite in the sense of Definition~\ref{def:periodisation}.
In order to provide motivation for the construction of a GRF with periodised covariance, it is
instructive to first consider the finite-dimensional case. It should be recalled that the transition
to this setting is achieved by the introduction of a grid $T_N \subset \R^d$. If this grid is uniformly
spaced, then the periodicity of the function $\prdi{\varphi}{\alpha}$ implies that the associated
covariance matrix $\prdi{C}{\alpha}$ as per \eqref{eq:componentwiseIsotropicCovarianceMatrix}, is 
Toeplitz and can be embedded in a circulant matrix (see \cite{lord2014introduction} Example 6.47 and
Definition 6.48).
Given that circulant matrices are diagonalised by a FFT and that periodisations of covariance matrices are
positive definite, this implies a factorisation of the covariance matrix of the form
$\prdi{C}{\alpha} = W \Lambda \adjnt{W}$, where $\Lambda$ is a diagonal matrix with positive entries on the
diagonal, and $W \in \R^{N \times N}$ denotes the $d$-dimensional Discrete Fourier Transform (DFT) matrix. 
By Definition~\ref{def:periodisation}, the entry $\lambda_\mu$ of $\Lambda$ associated to the grid location
$\mu \in \Z_n^d$ is given by $\lambda_\mu = \fourier{\varphi}((2 \alpha)^{-1} \mu)$.

Similarly to \eqref{eq:covFactorisation}, this factorisation can be used to sample realisations $z$ of a
centred multivariate Gaussian with covariance matrix $\prdi{C}{\alpha}$.
To this end, denote by $\xi$ a centred, complex-valued Gaussian random vector and define
\begin{equation}    \label{eq:samplingCirculantCovariance}
    z \coloneqq W \Lambda^{\nicefrac{1}{2}} \xi, \quad \text{where} \quad
        \exv \xi \xi^T = 0 \quad \text{and} \quad \exv \xi \adjnt{\xi} = 2 \, \mathbb{I}_{N \times N}.
\end{equation}
It can be shown (see e.g.~\cite[Lemma~6.51]{lord2014introduction}) that the real part $\mf{Re}(z)$
and the imaginary part $\mf{Im}(z)$ of $z$ are uncorrelated and each has covariance matrix $\prdi{C}{\alpha}$.
Note that if the matrix $\Lambda \in \R^{N \times N}$ is known, then computing $z$ in
\eqref{eq:samplingCirculantCovariance} essentially consists in drawing a realisation of $\xi$ and computing
the matrix-vector product with $W$, which can be achieved by a FFT with $\mc{O}(N \log N)$ computational
cost. We extend this sampling idea to the infinite-dimensional setting with the definition of \emph{periodised}
GRFs.
\begin{definition}[Periodised Random Field] \label{def:periodisedRF}
    To each positive definite periodisation $\prdi{\varphi}{\alpha}$ of a stationary covariance 
    function $\varphi$, we associate a \emph{periodised random field} $\prdi{u}{\alpha}$ through
    \begin{equation}    \label{eq:periRF}
        \prdi{u}{\alpha}(x, \omega) \coloneqq \alpha^{-\nicefrac{d}{2}}
            \suml{\mu \in \Z^d} \xi_\mu(\omega)
            \sqrt{\farg{\varphi}{\alpha^{-1} \mu}}~
            e^{-2 \pi i \alpha^{-1} \mu \cdot x}, \quad x \in \rpow{d},
    \end{equation}
    where the $\xi_\mu$, $\mu \in \Z^d$ are centred complex-valued Gaussian random
    vectors, with 
    \begin{equation}    \label{eq:canRFCoeff}
    \xi_{-\mu} = \conj{\xi_\mu}, 
        \quad \text{and} \quad \exv \xi_{\mu}\conj{\xi_\eta} = \delta_{\mu \eta},
        \quad \text{for } \mu, \eta \in \Z^d.
    \end{equation}
\end{definition}
Given Assumption~\ref{assump:statCov}, absolute convergence
of the series for any $x \in \R^d$ is assured by the PSF.
An elementary calculation then shows that $\exv \prdi{u}{\alpha}(x,\,\cdot\,) = 0$
for all $x \in \R^d$. In order to see that the covariance function of the periodised
random field is indeed a periodisation of the pristine covariance, recall from
\eqref{eq:hyperoctant} and \eqref{eq:hyperoctantDecomposition} that $\Z^d$ may be represented
as the disjoint union of the $H_q$, where $q \in \{-1, 1\}^d$.
Using this fact, along with the specific choice of stochastic coefficients in \eqref{eq:canRFCoeff} 
and rotational symmetry of isotropic covariances, we see that
\begin{align}
    \exv \prdi{u}{\alpha}(x, \,\cdot\,)~\prdi{u}{\alpha}(y, \,\cdot\,)
        &= \alpha^{-d} \!\suml{\mu, \eta \in \Z^d} \big(\!\exv \xi_\mu \xi_\eta\big)
            \sqrt{\fourier{\varphi}(\alpha^{-1} \mu)}\sqrt{\fourier{\varphi}(\alpha^{-1} \eta)}~
            e^{-2 \pi i \alpha^{-1}(\mu x + \eta y)} \\
        &= \alpha^{-d} \!\suml{\substack{q \in \{-1, 1\}^d\!, \\ q_1 = 1}} \suml{\mu \in H_{q}} 
            \big(\!\exv \xi_\mu \xi_{-\mu} \big)~
            \fourier{\varphi}(\alpha^{-1} \mu)~e^{-2 \pi i \alpha^{-1} \mu \cdot (x - y)} \\
        &= \alpha^{-d} \!\suml{q \in \{-1, 1\}^d} \suml{\mu \in H_{q}} \fourier{\varphi}(\alpha^{-1} \mu)~
            e^{-2 \pi i \alpha^{-1} \mu \cdot (x - y)} \\ 
        &= \prdi{\varphi}{\alpha}(x - y).
\end{align}
For practical purposes, we consider truncations of the series in \eqref{eq:periRF}.
Although fairly general truncations are possible in this framework, for ease of exposition we
restrict ourselves to truncations of the form  $\norm{\mu}_{\infty} \leq n$, for a given $n \in \N$.
We denote the corresponding truncated periodised random field by $\prdi{u}{\alpha, n}$ and 
its covariance function by $\prdi{\varphi}{\alpha, n}$. We note that after truncation,
identity \eqref{eq:PSFstationaryCovariance} no longer holds true. In other words, for the
truncated lattice $\Z_n^d$, in general it holds that $\prdi{\varphi}{\alpha, n}(\delta) \neq
\sum_{\eta \in \Z_n^d} \varphi(\delta + \alpha \eta)$ for $\delta \in \R^n$. This only becomes an equality in
the limit as $n$ tends to infinity.
In light of \eqref{eq:samplingCirculantCovariance}, any truncated periodised random field
$\prdi{u}{\alpha, n}$ can be evaluated on a uniformly spaced grid $T_N$ using a $d$-dimensional FFT,
which provides an efficient way to generate realisations on a computer.

At this point, it is worth noting how CE fits into this framework. To this end,
we replace $\rho$ in \eqref{eq:statCovDef} with $\chi \rho$, where $\chi$ is the indicator function for
$[-\alpha, \alpha]$ and consider a periodisation $\prdi{\varphi}{\alpha, \mrm{CE}}$
of $\varphi_{\mrm{CE}}(\delta) \coloneqq (\rho \chi)
(\ell^{-1} \norm{\delta}_2)$. Assuming this periodisation to be positive definite for the moment, 
gives
\begin{equation}    \label{eq:circulantEmbeddingPeriodisation}
    \prdi{\varphi}{\alpha, \mrm{CE}}(\delta) =
        \suml{\eta \in \Z^d} \big(\rho \chi\big) (\ell^{-1} \norm{\delta + \alpha \eta}_2)
        = \rho(\ell^{-1} \norm{\delta}_2), \quad \text{if} \quad \delta \in [-\alpha, \alpha]^d.
\end{equation}
That is, we recover not just a periodisation, but the desired covariance $\rho$ itself.
The CE method is equivalent to performing the usual transition to
the finite-dimensional setting via truncation of the series and evaluation on a uniformly spaced grid
for $\prdi{\varphi}{\alpha, \mrm{CE}}$.
Sampling Gaussian vectors $z$ with the corresponding circulant covariance matrix is then performed 
exactly as in \eqref{eq:samplingCirculantCovariance}, by using a FFT. As a consequence of 
\eqref{eq:circulantEmbeddingPeriodisation}, this Gaussian vector $z$ will have the desired
covariance matrix $C$, as long as positive definiteness of $\varphi_{\mrm{CE}}$ can be ensured.
However, although the periodisation of $\varphi$ is positive definite according
to Bochner's theorem, the periodisation of $\varphi_{\mrm{CE}}$ need not be positive definite
(barring convexity requirements on $\rho$, see \cite{wood1995truncated}). The authors in \cite{graham2018analysis}
have shown that in the \Matern case it is possible to recover positive
definiteness of $\prdi{\varphi}{\alpha, \mrm{CE}}$, provided $\alpha$ is chosen sufficiently large.
This result demonstrates that the required lower bound on $\alpha$ increases with both $\nu$ and $\ell$.

Evaluating a periodised random field, associated to either $\prdi{\varphi}{\alpha}$ or
$\prdi{\varphi}{\alpha, \mrm{CE}}$ using the FFT, always yields a uniform grid of
evaluations in $[-\alpha, \alpha]^d$. As we are only interested in random field
evaluations in $D = [0, 1]^d$, this means that any computed evaluation outside $D$ 
will be discarded. Thus, large values of $\alpha$ are inherently linked to
computational overhead. As $\alpha$ is chosen with respect to a given domain of interest,
this is independent of the specific choice of $D$.
Thus, for large smoothness $\nu$ and/or large correlation
length $\ell$, the choice of $\alpha$, necessary to 
recover positive definiteness of $\prdi{\varphi}{\alpha, \mrm{CE}}$ for \Matern covariance
may prove prohibitively large in practice\cite{graham2018analysis}. In contrast, periodisations as per
Definition~\ref{def:periodisation},
are always positive definite and so evaluations of the periodised random field associated to $\prdi{\varphi}{\alpha}$
are always possible -- independent of $\alpha$. In this sense, the choice of $\alpha$ provides a trade-off between the
accuracy of the sampled vector's covariance on the one hand, and the computational cost on the other.

    \subsection{SPDE-Based Sampling}    \label{sec:prel:spdeSampling}
        An alternative sampling approach, which has recently gained a lot of traction,
relies on the observation in \cite{whittle1954stationary, whittle1963stochastic}, that a GRF $u(\cdot, \omega)$
on $\R^d$ with \Matern covariance function satisfies
\begin{equation}    \label{eq:SPDEProblem}
    \big(\kappa^2 - \Delta\big)^{\beta} u(\cdot, \omega) = \mc{W}(\omega), \quad \text{in } \R^d,
        \quad \text{a.s.}
\end{equation}
Here, $\Delta$ denotes the Laplace operator, $\mc{W}$ represents $L^2$ white noise
in the sense of an isonormal Gaussian process (see \cite{nualart2018introduction}, Definition~1.1.1),
$\kappa \coloneqq \sqrt{2\nu} \ell^{\,-1}$ and $\beta = \nicefrac{\nu}{2} + \nicefrac{d}{4}$. The authors in
\cite{lindgren2011explicit,simpson2016going} propose
an approximation to the above SPDE problem on bounded domains $D \subset \R^d$ through the Finite-Element
method, where the problem on $D$ is augmented with homogeneous Neumann boundary conditions.
For $\beta = 1$, the FE method provides a finite-dimensional approximation $u_h(\cdot, \omega)$ to the
solution of the SPDE \eqref{eq:SPDEProblem} for almost every $\omega \in \Omega$, given a suitable
finite-dimensional representation $\xi^{(\mc{W})} \in \ltwo{D}$ of the white noise process $\mc{W}$.
This approximation is defined on a mesh of the target domain $D$ with maximum cell diameter $h > 0$.
Its key advantage in computing $u_h$ is the sparse representation of the factor matrix $H$ in
\eqref{eq:covFactorisation}, enabling the use of fast, scalable sparse solvers for generating realisations.
However, this comes at the cost of an approximation error, which can be reduced at an algebraic rate
by decreasing $h$ \cite{lindgren2011explicit,simpson2016going, cox2020regularity}.

Computing the FE approximation relies on a weak formulation of problem \eqref{eq:SPDEProblem}
that incorporates boundary conditions on $\partial D$. This reduces the computation of $u_h$ to projecting
onto a finite-dimensional subspace spanned by selected test functions and solving a linear system within
that subspace. Consequently, understanding the action of $\mc{W}$ and $L^{\beta}$ on these test functions
is necessary.
For any finite set of test function $\{v_i\}_{i = 1}^n \subset \ltwo{D}$ the actions $\mc{W}(v_i)$ are
jointly Gaussian, with expectation and covariance given by
\begin{equation}    \label{eq:whiteNoiseDef}
    \exv \mc{W}(v_i) = 0,   \quad \text{and} \quad \exv \mc{W}(v_i) \mc{W}(v_j) = \iprod{v_i, v_j},
        \quad \text{for } i, j = 1, \ldots, n.
\end{equation}
Defining $L \coloneqq \kappa^2 - \Delta$, the action of the fractional-power operator $L^{\beta}$ is
defined in terms of its spectral properties. More precisely, given an $\ltwo{D}$-orthonormal set of
eigenfunctions $\{e_j\}_{j \in \N}$ of $L$, associated to the eigenvalues $\{ \lambda_j\}_{j \in \N}$,
the action of $L^{\beta}$ on the domain
$\mc{D}(L^\beta) \coloneqq \big\{v \in \ltwo{D}
: \sum_{j \in \N} \lambda_j^{2 \beta} \iprod{v, e_j}^2 < \infty \big\} \subset \ltwo{D}$,
is given by
\begin{equation}    \label{eq:fracPowerAction}
    L^{\beta} v \coloneqq \suml{j \in \N} \lambda_j^{\beta} \iprod{v, e_j} e_j,
    \quad v \in \mc{D}(L^{\beta}).
\end{equation}
A more detailed and rigorous discussion of the weak formulation, the well-posedness and the regularity of
problem \eqref{eq:SPDEProblem} is available in \cite{cox2020regularity}. For incorporating arbitrary
exponents $\beta > 0$ within a FE framework, see \cite{bolin2020rational, bolin2020numerical}.
Efficient computation of approximations to $\mc{W}$ with respect to a FE basis is addressed in \cite{croci2018efficient}.

\section{Dirichlet-Neumann Averaging}
    \label{sec:sampFr}

    When evaluating a truncated GRF using a FFT, we implicitly impose 
periodic boundary conditions
on the random field. Our first aim in this section is to show that in a similar
vein, choosing
a specific combination of boundary conditions can
be used to provide an efficient sampling algorithm for isotropic GRFs.
To simplify the presentation and avoid the added technical and notational overhead
in general $d$ dimensions, we start by illustrating the idea in one dimension.
We will then generalise it to $d \geq 1$ dimensions in 
Definition~\ref{def:combRF}. In Proposition~\ref{prop:combRFCov} we prove that this
general construction indeed transfers the initial idea to arbitrary dimensions.
Lastly, we provide an error estimate for the special case of \Matern covariance
in Proposition~\ref{prop:covErrorMatern}.

Consider the open domain $D = (0, 1) \subset \R$ on which we want to generate GRF realisations and
allow for an extension $D_{\alpha} = (0, \alpha)$ to control the periodisation,
analogously to Section~\ref{sec:prel:prd}. For an expansion satisfying homogeneous
Neumann boundary conditions on $\partial D_\alpha$, a canonical basis is given by
appropriately scaled cosine functions. The analogous statement holds for homogeneous Dirichlet
boundary conditions and sine functions.  Following the same periodisation idea as in Definition~\ref{def:periodisedRF}, but under the constraint of these boundary conditions,
naturally leads to two GRFs, $\nrf{u}{\alpha}(x)$ and $\drf{u}{\alpha}$ defined
through
\begin{equation}    \label{eq:nrf1d}
    \nrf{u}{\alpha}(x) \coloneqq \xi_0
        \sqrt{2 \alpha^{-1} \farg{\varphi}{0}}
        + \sqrt{2 \alpha^{-1}} \suml{m \in \N}
            \xi_m \sqrt{\pcoeff{m}}~\cos(\sfrac{\pi m x}{\alpha}),
            \quad x \in \R,
\end{equation}
and\vspace{-2ex}
\begin{equation}    \label{eq:drf1d}
    \drf{u}{\alpha}(x) \coloneqq \sqrt{2 \alpha^{-1}} \suml{m \in \N}
        \xi_m \sqrt{\pcoeff{m}}~\sin(\sfrac{\pi m x}{\alpha}),
            \quad x \in \R,
\end{equation}
where\vspace{-2ex}
\begin{equation}   
    \xi_m \sim \stdnrm{0, 1} \quad \text{and} \quad \exv \xi_m \xi_k = \delta_{mk},
         \quad m, k \in \N_0.
 \end{equation}
These random fields correspond to a choice of homogeneous Neumann- and Dirichlet
boundary conditions on $\partial D_\alpha$ respectively. By construction,
$\exv \nrf{u}{\alpha} \equiv 0$
and $\exv \drf{u}{\alpha} \equiv 0$. Furthermore, using the definition of the $\xi_m$,
rotational symmetry of $\varphi$, and Euler's formula, we see that for any $x, y \in \R$,
\begin{align} \nonumber
    \exv \nrf{u}{\alpha}(x) \nrf{u}{\alpha}(y)
        &= \frac{2 \farg{\varphi}{0}}{\alpha} + \frac{2}{\alpha} \suml{m \in \N}
            \pcoeff{m} \cos(\sfrac{\pi m x}{\alpha}) \cos(\sfrac{\pi m y}{\alpha}) \\ \nonumber
        &= \frac{2 \farg{\varphi}{0}}{\alpha} + \frac{2}{\alpha} \suml{m \in \N}
            \pcoeff{m} \Big[
            \cos(\sfrac{\pi m}{\alpha} (x - y))
                + \cos(\sfrac{\pi m}{\alpha} (x + y))\Big] \\ \nonumber
        &= \frac{1}{\alpha} \suml{m \in \Z} \! \pcoeff{m} e^{i \pi \alpha^{-1} m (x - y)} 
            + \frac{1}{\alpha} \suml{m \in \Z} \! \pcoeff{m} e^{i \pi \alpha^{-1} m (x + y)}  \\
        &= \prdi{\varphi}{2 \alpha}(x - y)
            + \prdi{\varphi}{2 \alpha}(x + y).     \label{eq:neuCov}
\end{align}
An analogous calculation for $\drf{u}{\alpha}$ yields
\begin{equation}    \label{eq:dirCov}
    \exv \drf{u}{\alpha}(x) \drf{u}{\alpha}(y)
        = \prdi{\varphi}{2 \alpha}(x - y) - \prdi{\varphi}{2 \alpha}(x + y).
\end{equation}
In both cases, the covariance function of the GRFs can be expressed in terms of a periodisation of the pristine
covariance, with two crucial alterations compared to periodic boundary conditions. Firstly, the scaling parameter
$\alpha$ of the periodisation is twice as large in comparison to the corresponding periodisation with periodic
boundary conditions on the same domain $D_{\alpha}$. Secondly, the presence of terms containing
$\prdi{\varphi}{2 \alpha}(x + y)$ results in the loss of
isotropy for both covariance functions. However, this anisotropy is confined to the boundaries and diminishes
towards the centre of the domain roughly as
$\prdi{\varphi}{2 \alpha}(2 r) \geq \rho(2 \ell^{-1} |r|)$, where $r$ denotes the distance from the boundary.
This is consistent with the findings in \cite{khristenko2019analysis} Theorem~3.2, where the authors proved an
analogous upper bound for the covariance error in general dimension $d$.

    A key observation from (\ref{eq:neuCov}-\ref{eq:dirCov}) is that the terms that break 
the isotropy at the boundary have exactly opposite signs. Thus, the Dirichlet-Neumann
averaged GRF
\begin{equation}    \label{eq:combRF1d}
    u_{\alpha}(x) \coloneqq \frac{1}{\sqrt{2}} \Big(\nrf{u}{\alpha}(x) + \drf{u}{\alpha}(x)\Big),
        \quad x \in \R,
\end{equation}
where $\nrf{u}{\alpha}$ and $\drf{u}{\alpha}$ are independent,  satisfies
\begin{equation}
    \exv u_{\alpha}(x) = 0 \quad \text{and} \quad
        \exv u_{\alpha}(x) u_{\alpha}(y) = \prdi{\varphi}{2 \alpha}(x - y),
            \quad x, y \in \R,
\end{equation}
and hence is again isotropic.
In conclusion, the covariance function of the averaged random field \eqref{eq:combRF1d}
is a periodisation of the pristine covariance function with
twice the scaling parameter of the random field with 
periodic boundary conditions in \eqref{eq:periRF}.

Analogous to the evaluation 
of \eqref{eq:periRF} using a FFT, the GRFs $\nrf{u}{\alpha}$ and $\drf{u}{\alpha}$ can 
be evaluated on a uniformly spaced grid using a DCT and a DST respectively. These
transforms directly yield evaluations on the desired domain $D_\alpha$, in contrast
to the FFT, which yields evaluations in $(-\alpha, \alpha)$ and where
consequently half of the resulting vector will be discarded. Moreover, as the
DCT and DST only take real-valued coefficients, the number of floating-point values
required is half that needed in the FFT.
Hence, the cost of evaluating the periodised random field \eqref{eq:periRF} and that of the
DNA random field in \eqref{eq:combRF1d} are comparable, even when accounting for the
fact that the real and imaginary parts of the FFT provide two independent realisations.
At the same time, the fact that the scaling parameter is doubled, implies that the error is 
significantly reduced. The extent to which it is reduced in the case of \Matern
covariance is shown in Proposition~\ref{prop:covErrorMatern} below, as well as 
numerically verified in Section~\ref{sec:numExp} for other covariance functions.

The cost comparison with CE is similar. The smallest possible embedding for CE,
referred to as 'minimal embedding' in \cite{dietrich1997fast},
corresponds to periodic boundary conditions with $\alpha = 1$. In 
practice however, significantly larger values of $\alpha$
may be required. Although there is no error through
periodisation, the considerably higher computational effort and required fine-tuning
of $\alpha$ in CE (i.e.~determining the padding) might still make it worthwhile to accept the trade-off between error and
speed.

\subsection{Generalisation to Higher Dimensions and Error Estimates}     \label{sec:higherDim}

As we are ultimately interested in evaluations on $D_{\alpha} = (0, \alpha)^d$,
the natural way to generalise the idea presented in the previous section is
to consider a tensor product of the one-dimensional construction. More specifically,
consider a partition of the boundary $\partial D_{\alpha} = \bigcup_{j = 1}^{\,d}
\Gamma_{j, \alpha}$, where
    \begin{equation}
        \Gamma_{j, \alpha} \coloneqq \big\{ x = (x_1, \ldots, x_d)
            \in \partial D_{\alpha}: x_j \in \{0, \alpha\}\big\}, \quad j = 1, \ldots, d.
    \end{equation}
We define the averaged random field in $d$ dimensions as the sum of all $2^d$ 
combinations of random fields with homogeneous Neumann and homogeneous Dirichlet 
boundary conditions on the $\Gamma_{j, \alpha}$, $j = 1, \ldots, d$ respectively.
Making this construction explicit requires some notational machinery, which we want 
to develop in the following.

We indicate the use of homogeneous Neumann boundary conditions on a 
$\Gamma_{j, \alpha}$ by the index $b_j = 0$ and correspondingly use $b_j = 1$ for
homogeneous Dirichlet boundary conditions. A particular choice of boundary
conditions on $D_{\alpha}$ can then be encoded in a vector $b \in \{0, 1\}^d$.
Since the expansion for Dirichlet boundary conditions does not contain the 
constant term, we further introduce the index sets
\begin{equation}    \label{eq:truncatedIndexSets}
    \mc{I}_{b, n} \coloneqq \mc{I}_{b_1, n} \times \ldots \times \mc{I}_{b_d, n},
        \ \ \text{where} \ \ \mc{I}_{0, n} \coloneqq \{ 0, 1, \ldots, n \}
        \ \ \text{and} \ \ 
        \mc{I}_{1, n} \coloneqq \mc{I}_{0, n} \setminus \{0\}.
\end{equation}
With this notation at hand, we define the basis functions
$e_{\mu, \alpha}^b: \R^d \to \R$, as
\begin{equation}    \label{eq:perRFModes}
    e_{\mu, \alpha}^b (x) \coloneqq
        \left(\frac{2}{\alpha}\right)^{\nicefrac{d}{2}} \prodlim{j = 1}{d}
        e_{\mu_j}^{(b_j)} (\nicefrac{x_j}{\alpha}),
        \quad \mu \in \mc{I}_{b, n},
\end{equation}
where\vspace{-2ex}
\begin{equation}
        e_{\mu_j}^{(0)}(\theta) \coloneqq \cos (\pi \mu_j \theta), 
        \quad \text{and} \quad e_{\mu_j}^{(1)} (\theta) \coloneqq \sin (\pi \mu_j \theta),    \quad \theta \in \R.
\end{equation}
With this, the basic building blocks are in place and we are able to 
define the averaged random field in d dimensions.
\begin{definition}[Dirichlet-Neumann Averaged Gaussian Random Field]   \label{def:combRF}
    Let the independent GRFs $u_{\alpha, n}^{b}: D_{\alpha} \times \Omega \to \R$ be defined through
    \begin{equation}    \label{eq:combRFBndry}
        u_{\alpha, n}^{b} (x, \omega) \coloneqq \suml{\mu \in \mc{I}_{b, n}} \xi_\mu (\omega)
            \sqrt{\farg{\varphi}{(2\alpha)^{-1} \mu}}~e_{\mu, \alpha}^{b}(x),
                \quad b \in \{0, 1\}^d,
    \end{equation}
    with
    $\xi_{\mu} \sim \stdnrm{0, 1}$, such that $\exv \xi_{\mu} \xi_{\eta} = \delta_{\mu \eta}$,
    for all $\mu, \eta \in \N_0^d$.
    Then, we refer to
    \begin{equation}    \label{eq:combinedRandomField}
        u_{\alpha, n} (x, \omega) \coloneqq 
            2^{-\nicefrac{d}{2}} \suml{b \in \{0, 1\}^d}
            u_{\alpha, n}^{b}(x, \omega),
    \end{equation}
    as the Dirichlet-Neumann Averaged Gaussian Random Field (DNA GRF).
\end{definition}

We are now in a position to verify that the DNA GRF is indeed
isotropic and the scale parameter $\alpha$ of the associated periodisation
is two times larger than for a naive periodisation.
\begin{proposition}   \label{prop:combRFCov}
    The DNA GRF $u_{\alpha, n}$ in $d$ dimensions satisfies
    \begin{equation}
        \exv u_{\alpha, n}(x, \,\cdot\,) = 0, \quad \text{and} \quad
            \exv u_{\alpha, n}(x, \,\cdot\,) u_{\alpha, n}(y, \,\cdot\,)
            = \prdi{\varphi}{2 \alpha, n}\big(x - y),
            \quad \text{for all } x, y \in \rpow{d}.
    \end{equation}
\end{proposition}
\begin{proof}
    Throughout the proof, we will denote
        $\lambda_{\alpha, \mu} \coloneqq \fourier{\varphi}\big((2 \alpha)^{-1} \mu\big)$
        for $\alpha \geq 1,~\mu \in \Z^d$.
    The fact that $\exv u_{\alpha, n}(x, \cdot) = 0$ for all $x \in \R$ follows immediately from
    linearity of expectation and the properties of the stochastic coefficients $\xi_{\mu}$.
    
    For any choice of boundary conditions encoded in $b \in \{0, 1\}^d$,
    again due to the properties of the stochastic coefficients, the covariance
    function of $u^{b}_{\alpha, n}$ is given by
    \begin{equation}    \label{eq:unmodified00}
        \exv u^{b}_{\alpha, n}(x, \cdot) \, u^{b}_{\alpha, n}(y, \cdot)
            =  \left(\frac{2}{\alpha}\right)^{d}\suml{\mu \in \mc{I}_{b, n}^{d}} \lambda_{\alpha, \mu}
                \prodlim{j = 1}{d} \e{b_j}{\mu_j}(\alpha^{-1} x_j)~\e{b_j}{\mu_j}(\alpha^{-1} y_j).
    \end{equation}
    Now, 
    standard trigonometric identities imply that for any $\tau_1, \tau_2 \in \R$
    and $r \in \{0, 1\}$,
    \begin{equation}
        \label{eq:trigIdBasis}
        \e{r}{k}(\tau_1) \e{r}{k}(\tau_2) = \frac{1}{2} \Big(\e{0}{k}(\tau_1 - \tau_2)
            + (-1)^{r} \e{0}{k}(\tau_1 + \tau_2)\Big), \quad k \in \N_0.
    \end{equation}
     Using this identity, 
     the covariance function 
     can for any $ b \in \{0, 1\}^d$
    be rewritten as
    \begin{align}
    \exv u^{b}_{\alpha, n}(x, \cdot) \, u^{b}_{\alpha, n}(y, \cdot)
        = \alpha^{-d} \sum_{\mu \in \mc{I}_{b, n}^{d}} \lambda_{\alpha, \mu}
            \prod_{j = 1}^{d} \Big[&\e{0}{\mu_j}\big(\alpha^{-1}(x_j - y_j)\big) \nonumber \\
            &+ (-1)^{b_j} \e{0}{\mu_j}\big(\alpha^{-1}(x_j + y_j)\big)\Big].
            \label{eq:intermediateCov}
\end{align}
    A straight-forward induction
    on $d$, shows that for any real sequences $\seq{\alpha_j}{j}$, $\seq{\beta_j}{j}$ and $d \in \N$ it
    holds that
    \begin{equation}    \label{eq:cancellation}
        \suml{b \in \{0, 1\}^d} \prodlim{j = 1}{d} \left(\alpha_j + (-1)^{b_j} \beta_j \right)
            = 2^d \prodlim{j = 1}{d} \alpha_j.
    \end{equation}
    As a consequence of \eqref{eq:cancellation} and the independence of the $u_{\alpha,n}^b$, the covariance
    \eqref{eq:unmodified00} of the averaged random field $u_{\alpha, n}$ from \eqref{eq:combinedRandomField}
    reduces to
    \begin{align}
        \exv u_{\alpha, n} (x, \cdot) \, u_{\alpha, n}(y, \cdot)
            &= 2^{-d} \suml{b \in \{0, 1\}^d} \exv~u_{\alpha, n}^b (x, \cdot) \, u_{\alpha, n}^b(y, \cdot) \\
            &= (2 \alpha)^{-d} \suml{\mu \in \N_n^d} \lambda_{\alpha, \mu} 
                \suml{q \in \{-1, 1\}^d} \cos\bigg(\pi \alpha^{-1} \sumlim{j = 1}{d}  q_j \mu_j (x_j - y_j)\bigg).
    \end{align}
    where we have used the definitions of $\mc{I}_{b, n}^{d}$ and $\e{0}{\mu_j}$ and
    the trigonometric identity
    \begin{equation}
        \prodlim{j = 1}{d} \cos (\vartheta_j)
            = 2^{-d} \suml{q \in \{-1, 1\}^d} \cos \big(\sumlim{j = 1}{d} q_j \vartheta_j\big),
            \quad \text{for any } \vartheta \in \R^d.
    \end{equation}
    Recalling \eqref{eq:hyperoctantDecomposition} we can finally see\vspace{-2ex}
    \begin{align}
        \exv u_{\alpha, n} (x) u_{\alpha, n}(y)
            &= (2 \alpha)^{-d} \suml{q \in \{-1, 1\}^d}\suml{\mu \in H_{q, n}} \lambda_{\alpha, \mu} 
                \cos\bigg(\pi \alpha^{-1} \sumlim{j = 1}{d} \mu_j (x_j - y_j)\bigg) \\
            \label{eq:rotSym1}
            &= (2 \alpha)^{-d} \suml{\mu \in \Z_n^d} \lambda_{\alpha, \mu} 
                \cos\big(\pi \alpha^{-1} \mu \cdot (x - y) \big) \\
            \label{eq:covRepresentationTruncated}
            &= (2 \alpha)^{-d} \suml{\mu \in \Z_n^d} \lambda_{\alpha, \mu}~ 
                e^{i \pi \alpha^{-1} \mu \cdot (x - y)} \\
            &= \prdi{\varphi}{2 \alpha, n} (x - y),
    \end{align}
    where we have utilised the rotational symmetry of $\fourier{\varphi}$ from 
    \eqref{eq:rotSym1} to \eqref{eq:covRepresentationTruncated}.
\end{proof}

We have previously argued that the DNA sampling approach provides
an efficient way to generate GRF realisations. However, unlike CE, the method
produces a covariance that is not exactly the desired 
covariance, but a periodisation of it. Therefore, the accuracy and 
applicability of this method depend on the error committed.

We investigate the maximal point-wise error
\begin{equation}    \label{eq:covErrorDef}
    \norm{\prdi{\varphi}{2 \alpha, n} - \varphi}_{\infty} \coloneqq
    \max\limits_{\delta \in [-1, 1]^d} \big| \, \prdi{\varphi}{2 \alpha, n}(\delta) - \varphi(\delta) \, |
\end{equation}
between the covariance function $\prdi{\varphi}{2 \alpha, n}(\delta)$ of the averaged random 
field as per Proposition~\ref{prop:combRFCov} and the desired, isotropic
covariance $\varphi(\delta) = \rho\big(\ell^{-1} \norm{\delta}_2\big)$. The
covariance error \eqref{eq:covErrorDef} naturally splits into a component
related to the truncation of the corresponding random field and a component due to periodisation,
\begin{equation}
    \norm{\varphi - \prdi{\varphi}{2 \alpha, n}}_{\infty} \leq 
        \norm{\varphi - \prdi{\varphi}{2 \alpha}}_{\infty}
        + \norm{\prdi{\varphi}{2 \alpha} - \prdi{\varphi}{2 \alpha, n}}_\infty.
\end{equation}
To derive explicit rates of convergence with respect to $\alpha$ and $n$, we consider in the
following the special class of \Matern covariances with correlation length
$\ell > 0$ and smoothness parameter $\nu > 0$, as introduced in \eqref{eq:maternDefinition}.
We first turn our attention to the error contribution due to periodisation of the desired covariance.
\begin{lemma}   \label{lem:periodError}
    Assume that $\nu \geq \nicefrac{1}{2}$ and that $\alpha \geq 1$ is chosen such that
    $2 \alpha \kappa > \max \{\nicefrac{3}{2}, d + \nu - \nicefrac{3}{2} \}$. Then,
    \begin{equation}    \label{eq:covErrorMatern}
    \norm{\varphi - \prdi{\varphi}{2 \alpha}}_{\infty}
        \leq 
            C_1~e^{-2 \vartheta \alpha \kappa}
            \bigg((\alpha \kappa)^{\nu - \nicefrac{1}{2}}
            + \frac{\Gamma(d + \nu + \nicefrac{1}{2})}{(\alpha \kappa)^{1 + d}}\bigg),
    \end{equation}
where
    $0 < \vartheta \coloneqq
    \Gamma(d + \nu + \nicefrac{1}{2})^{-\frac{1}{d + \nu - \nicefrac{1}{2}}} < 1$,
    and $C_1 \coloneqq d 2^{d + 2 \nu - 1} e^{1 + \kappa}$.
\end{lemma}
\begin{proof}
    Since $\rho$ is monotonically decreasing away from the origin, we can bound
    \begin{align}
        \prdi{\varphi}{2 \alpha}(\delta) - \varphi(\delta)
            &= \suml{\eta \in \Z^d \setminus \{0\}} \rho\big(\ell^{-1} \norm{\delta + 2 \alpha \eta}_2\big) \\
            &\leq \suml{\eta \in \Z^d \setminus \{0\}} \rho\Big(\ell^{-1} \norm{\delta + 2 \alpha \eta}_\infty\Big) \\
            &\leq \suml{\eta \in \Z^d \setminus \{0\}}
                \rho\Big(\ell^{-1} \big(2 \alpha \norm{\eta}_\infty - 1\big)\Big),
    \end{align}
    where in the last bound we have used that $\norm{\delta + \mu}_{\infty} \geq \norm{\mu}_{\infty} - 1$ for any
    $\delta \in [-1, 1]^d$ and $\mu \in \Z^d$. Absolute convergence of the above series is assured by the PSF, so we
    can rearrange the summation over $\Z^d \setminus \{0\}$ in 
    terms of \enquote{shells} of $\mu \in \Z^d$ with a given $\infty$-norm, that is,
    \begin{equation}    \label{eq:shellSum}
        \suml{\eta \in \Z^d \setminus \{0\}}
            \rho\Big(\ell^{-1} \big(2 \alpha \norm{\eta}_{\infty} - 1\big)\Big)
            = \sumlim{m = 1}{\infty} \suml{\substack{\eta \in \Z^d \setminus \{0\}, \\ \norm{\eta}_\infty = m}}
                \rho \big(\ell^{-1} (2 \alpha m - 1)\big).
    \end{equation}
    Using the bound $\big|\{ \mu \in \Z^d : \norm{\mu}_{\infty} = m\}\big| \leq 2 d (2m)^{d - 1}$ and inserting
    the definition of the \Matern covariance \eqref{eq:maternDefinition} with $\kappa \coloneqq \nicefrac{\sqrt{2 \nu}}{\ell}$,
    we arrive at
    \begin{align}
        \prdi{\varphi}{2 \alpha}(\delta) - \varphi(\delta) 
            &\leq d 2^d \sumlim{m = 1}{\infty} m^{d - 1} \rho \big(\ell^{-1} (2 \alpha m - 1) \big) \\
            \label{eq:intermediateError}
            &= d \frac{2^d 2^{1 - \nu}}{\Gamma(\nu)} \sumlim{m = 1}{\infty}
                m^{d-1} \big(\kappa (2 \alpha m - 1)\big)^{\nu}~K_{\nu}\big(\kappa (2 \alpha m - 1)\big).
    \end{align}
    The assumption $2 \alpha \kappa \geq \nicefrac{3}{2}$ ensures that the smallest argument to 
    $K_{\nu}$ is larger than $\nicefrac{1}{2}$ . This implies that Lemma~2 from
    \cite{bachmayr2020unified} is applicable, which states that
    \begin{equation}
        K_{\nu}(t) \leq e \frac{2^{2 \nu - 1} \Gamma(\nu)}{\sqrt{2 t}} e^{-t},
        \quad \text{for } \nu > 0 \text{ and } t \geq \frac{1}{2}.
    \end{equation}
    Inserting this into \eqref{eq:intermediateError} yields
    \begin{equation}    \label{eq:intermediateError1}
        \prdi{\varphi}{2 \alpha}(\delta) - \varphi(\delta) \leq
            d 2^{d + 2\nu - 1} e^{1 + \kappa} (\alpha \kappa)^{\nu - \nicefrac{1}{2}}
               \sumlim{m=1}{\infty} m^{d + \nu - \nicefrac{3}{2}} e^{-2 \alpha \kappa m}. 
    \end{equation}
    Finally, the assumption $2 \alpha \kappa > d + \nu - \nicefrac{3}{2}$ 
    assures that we can apply Lemma~\ref{app:lem:auxSum} in the Appendix 
    with $\beta = d + \nu - \nicefrac{3}{2} < 2 \alpha \kappa = \gamma$.
    This results in the desired bound.
\end{proof}

Lemma~\ref{lem:periodError} provides a non-asymptotic upper bound for the periodisation error,
assuming $\alpha$ is sufficiently large to satisfy the lemma’s conditions. 
This assumption is typically no stricter than $\alpha \geq 1$.
Lemma~\ref{lem:periodError} shows that the
periodisation error decays asymptotically as $\mc{O}(\alpha^{\nu - \nicefrac{1}{2}} e^{-2 \vartheta \kappa \alpha})$,
that is, exponentially in $\alpha$. We now turn to the second error component, the truncation error from the averaged
random field series.

\begin{lemma}   \label{lem:truncError}
    The covariance function $\prdi{\varphi}{2 \alpha, n}$ of the DNA GRF
    \eqref{eq:combinedRandomField} satisfies
    \begin{equation}
        \norm{\prdi{\varphi}{2 \alpha, n} - \prdi{\varphi}{2 \alpha}}_\infty
            \leq C_2~(\alpha \kappa)^{2 \nu + d} n^{-2 \nu},
    \end{equation}
    where
    $C_2 \coloneqq d C_{\nu} (2 \kappa)^d \pi^{-(2 \nu + d)}
                \Big(1+\frac{1}{2 \nu}\Big)$ and $C_{\nu}$ is the constant in \eqref{eq:maternHankel}. 
\end{lemma}
\begin{proof}
    As per \eqref{eq:covRepresentationTruncated}, the covariance of the truncated random field is 
    given by
    \begin{equation}
        \prdi{\varphi}{2 \alpha, n}(\delta)
            = (2 \alpha)^{-d} \suml{\mu \in \Z_n^d}  \fourier{\varphi}\big((2 \alpha)^{-1} \mu\big)
                e^{i \pi \alpha^{-1} \mu \cdot \delta}, \quad \delta \in \R^d.
    \end{equation}
    Comparing this with the definition of the periodised covariance \eqref{eq:periodisationDefinition}, we
    see that
    \begin{equation}
        \big|~\prdi{\varphi}{2 \alpha, n} (\delta) - \prdi{\varphi}{2 \alpha}(\delta)~\big|
            = \left|\suml{\mu \in \Z^d \setminus \Z_n^d} \fourier{\varphi}\big((2 \alpha)^{-1} \mu\big)
                e^{i \pi \alpha^{-1} \mu \cdot \delta} \right|
            \leq \suml{\mu \in \Z^d \setminus \Z_n^d} \fourier{\varphi}\big((2 \alpha)^{-1} \mu\big),
    \end{equation}
    where we have used Bochner's theorem to guarantee positivity of $\fourier{\varphi}\big((2 \alpha)^{-1} \mu\big)$. 
    
    Next, define
    \begin{equation}
        f: [0, \infty) \to \R, \quad
            x \mapsto C_{\nu} \kappa^{-d}
                \left(1 + \frac{\pi^2}{\alpha^2 \kappa^2} x^2 \right)^{-(\nu + \nicefrac{d}{2})}.
    \end{equation}
    Recalling the Hankel transform \eqref{eq:maternHankel} of the \Matern covariance function, we see that
    this definition is motivated by the fact that
    $f(\norm{\mu}_2) = \fourier{\varphi}\big((2 \alpha)^{-1} \mu\big)$. Together with the observation
    that $f$ is monotonously decreasing and $\norm{\mu}_2 \geq \norm{\mu}_{\infty}$ for all
    $\mu \in \Z^d$, we can thus write
    \begin{equation}
        \suml{\mu \in \Z^d \setminus \Z_n^d} \!\! \fourier{\varphi}\big((2 \alpha)^{-1} \mu\big)
            = \suml{\mu \in \Z^d \setminus \Z_n^d} \!\! f\big(\norm{\mu}_2\big)
            \leq \suml{\mu \in \Z^d \setminus \Z_n^d} \!\! f\big(\norm{\mu}_{\infty}\big)
            = \sumlim{m = n+1}{\infty} \!\! \suml{\substack{\mu \in \Z^d, \\\norm{\mu}_{\infty} = m}} \! f(m).
    \end{equation}
    where we have arranged the summation over $\Z^d \setminus \Z_n^d$ in terms of \enquote{shells},
    analogously to the Proof of Lemma~\ref{lem:truncError}.
    Using the bound $\big|\{ \mu \in \Z^d : \norm{\mu}_{\infty} = m\}\big| \leq 2 d (2m)^{d - 1}$ and
    again the fact that $f$ is monotonously decreasing, we can employ a simple integral test, which
    yields
    \begin{equation}  \label{eq:intTestEstimate0}  
        \sumlim{m = n+1}{\infty} \! \suml{\substack{\mu \in \Z^d, \\\norm{\mu}_{\infty} = m}} \! f(m)
            \leq d 2^d \left((n+1)^{d-1} f(n+1) + \int\limits_{n+1}^{\infty}x^{d-1} f(x)~\mrm{d}x \right).
    \end{equation}
    For the first term we can immediately invoke the definition of the function $f$ and see
    \begin{equation}    \label{eq:intTestEstimate1}
        (n+1)^{d-1} f(n+1) \leq C_{\nu} \kappa^{-d}
            \left(\frac{\alpha \kappa}{\pi}\right)^{2 \nu + d} n^{-(1 + 2 \nu)}.
    \end{equation}
    Similarly, for the second term we get
    \begin{align}    
        \int\limits_{n+1}^{\infty} x^{d-1} f(x)~\mrm{d}x
            &\leq C_{\nu} \kappa^{-d}
                \left(\frac{\alpha \kappa}{\pi}\right)^{2 \nu + d} \int\limits_{n+1}^{\infty}x^{-(1+2\nu)}~\mrm{d}x
                \nonumber 
            \leq \frac{C_\nu \kappa^{-d}}{2\nu} \left(\frac{\alpha \kappa}{\pi}\right)^{2 \nu + d} n^{-2 \nu}.
            \label{eq:intTestEstimate2}
    \end{align}
    Combining this with \eqref{eq:intTestEstimate0} and \eqref{eq:intTestEstimate1} 
     finally gives the desired bound.
\end{proof}

The upper bound in Lemma~\ref{lem:truncError} grows with $\alpha$. Counteracting this requires
scaling the size $n$ of the truncation accordingly.
However, in the large-scale applications we consider, this does not pose a restriction. As with
Lemma~\ref{lem:periodError}, Lemma~\ref{lem:truncError} provides an explicit non-asymptotic upper bound. In
Proposition~\ref{prop:covErrorMatern} below, we combine these results to derive an asymptotic
upper bound on the total covariance error, which represents the central result of this section. The
proposition demonstrates that by increasing the value of the scaling parameter $\alpha$ and scaling
the value of $n$ appropriately, it is possible to achieve arbitrarily small values of the covariance error.
This theoretical justification will be underpinned by numerical experiments presented in Section~\ref{sec:numExp},
where it will be shown that even for moderate $n$ and the smallest value of $\alpha$, namely $\alpha = 1$,
errors are achieved which are negligible for many practical purposes.
\begin{proposition}   \label{prop:covErrorMatern}
    Assume that $\nu \geq \nicefrac{1}{2}$. Then there exists a constant
    $C > 0$ independent of $n \in \N$ and $\alpha \geq 1$, as well as values
    $n_0 \in \N$ and $\alpha_0 \geq 1$, such that
    \begin{equation}
        \norm{\prdi{\varphi}{2 \alpha, n} - \varphi}_{\infty} \leq
            C \Big( \alpha^{2 \nu + d} n^{-2 \nu} + \alpha^{\nu - \nicefrac{1}{2}}
            e^{-2 \vartheta \kappa \alpha} \Big),
    \end{equation}
    for all $n > n_0$ and $\alpha > \alpha_0$.
\end{proposition}
\begin{proof}
    From the fact that $\lim_{\alpha \to \infty}
    \Gamma (d + \nu + \nicefrac{1}{2}) \alpha^{-(d + \nu - \nicefrac{1}{2})} = 0$,
    we can infer the existence of a constant $a_1$ independent of $\alpha$ and a
    real number $\alpha_1 \geq 1$, such that
    $\Gamma(d + \nu + \nicefrac{1}{2})(\alpha \kappa)^{-(1 + d)}
            \leq a_1 (\alpha \kappa)^{\nu - \nicefrac{1}{2}}$,
    for all $\alpha > \alpha_0$.
    Choosing $\alpha_0 = \max \{\alpha_1, \alpha_2\}$, where $\alpha_2 \geq 1$ is chosen
    so as to satisfy the assumption of Lemma~\ref{lem:periodError}, that is,
    $2 \alpha_2 \kappa > \max \{\nicefrac{3}{2}, d + \nu - \nicefrac{3}{2}\}$, ensures the existence of $c_1 > 0$,
    such that
    \begin{equation}
        \norm{\varphi - \prdi{\varphi}{2 \alpha}}_{\infty}
        \leq c_1~\alpha^{\nu - \nicefrac{1}{2}} e^{-2 \vartheta \kappa \alpha}.
    \end{equation}
    The subsequent choice of $C = \max \{c_1, C_2\}$, together with the statement
    of Lemma~\ref{lem:truncError} then implies the proposition.
\end{proof}

\begin{remark}  \label{rem:boundingBox}
    The construction of the DNA GRF in Definition~\ref{def:combRF} assumes that the domain $D$ is a
    $d$-dimensional hypercube and a fast evaluation using DSTs and DCTs entails a uniform mesh. However, for
    applications where these prerequisites are not met, the target domain can be embedded in a hypercube and
    the evaluations can be performed on a uniform mesh. Evaluations at other, non-uniformly spaced points
    within the (embedded) target domain can then be obtained by interpolation from the uniformly-spaced values
    on the hypercube. For more details 
    see \cite{graham2018circulant, simpson2016going}.
\end{remark}

    \subsection{Connection to SPDE-Based Sampling} \label{sec:frame:spdeConnection}
        The previous section described a sampling framework that efficiently generates GRF
realisations for moderate to large random fields. However, the use of DSTs and DCTs
remains a limitation due to their necessity for regularly spaced grids and their
inherent challenges in massively parallel settings. In this regard, the SPDE
approach as described in Section~\ref{sec:prel:spdeSampling}, leverages the sparsity
and mature solution techniques within the FE framework in order to
achieve a scalable sampling algorithm. However, as with periodisation-based
approaches, restricting a random field to a bounded domain and imposing boundary
conditions introduces aliasing effects, which are a major obstacle to efficiency.

In \cite{lindgren2011explicit} the authors show that the covariance function
of a GRF with homogeneous Neumann boundary
conditions, sampled using the SPDE approach can be represented in terms of a
cosine-series (\cite{lindgren2011explicit},
Theorem 1). The authors refer to this covariance function as a
\emph{folded covariance}. Using a similar framework, the authors in
\cite{khristenko2019analysis} derive explicit error estimates for
such folded covariances in terms of different boundary conditions on the 
entire boundary. 
Both analyses prove that GRFs sampled with the SPDE approach using either 
homogeneous Neumann or Dirichlet boundary conditions exhibit a loss of isotropy that
is stronger the closer one is to the boundary. The suggested remedy is to compute the
solution on an extended domain and then truncate it back to the desired domain.
This approach, however, gives rise to challenges pertaining to memory,
implementation and scalability, which were previously alluded to in the
introduction.
Translating the ideas of Section~\ref{sec:sampFr} to the SPDE approach, we
propose replacing the extension of the domain by an averaging of several
solutions on the initial domain. This yields genuinely isotropic GRFs and addresses
the issues with using oversampling in the SPDE approach mentioned above.

Informally, Proposition~\ref{prop:SPDECorrespondence} below states that
in the limit $n\to \infty$ each of the GRFs $u_{\alpha, n}^b$ in \eqref{eq:combRFBndry} 
solves a particular, scaled SPDE on the hypercube $D_\alpha = (0, \alpha)^d$ with
the choice of the boundary conditions encoded in the index vector $b \in \{0, 1\}^d$,
as defined in Section~\ref{sec:higherDim}.
Therefore, in analogy to
\eqref{eq:combinedRandomField}, solutions to these SPDEs can be averaged 
to define a Dirichlet-Neumann averaged GRF 
for the SPDE approach. 
The DNA GRF has a periodised covariance function and
satisfies the explicit periodisation error bound in Lemma~\ref{lem:periodError}.
Thus, Proposition~\ref{prop:SPDECorrespondence} implies that the SPDE approach
of \cite{lindgren2011explicit} and the DCT/DST-based approach from
Section~\ref{sec:higherDim} can essentially be viewed as two different means of
discretising the same, infinite-dimensional GRF.

Throughout this section we will use the same notation as in
Section~\ref{sec:higherDim} to encode boundary conditions on
$\partial D_{\alpha} = \bigcup_{j = 1}^{d} \Gamma_{j, \alpha}$ in a vector
$b \in \{0, 1\}^d$. We are interested in weak solutions to the stochastic
boundary value problem (BVP)
\begin{subequations}    \label{eq:spdeBVP}
\begin{alignat}{2}
    \label{eq:spdeBVP_interior}
    \big(\kappa^2 - \Delta\big)^{\beta} u
        &= \sqrt{C_{\nu}} \kappa^{-\nu}~\mc{W} \quad &&\text{in } D_{\alpha} \quad \text{a.s.} \\
    \label{eq:spdeBVP_bndry}
    b_j~\partial_{\hat{n}} u + (1 - b_j)~u &= 0 \quad &&\text{on } \Gamma_{j, \alpha},
        \quad j = 1, \ldots, d,
\end{alignat}
\end{subequations}
where $\beta = \nicefrac{\nu}{2} + \nicefrac{d}{4}$, $\kappa = \sqrt{2 \nu} \ell^{\,-1}$ and
$C_{\nu}$ from \eqref{eq:maternHankel}. 
We define $u_{\alpha}^b$ in an analogous manner to
$u_{\alpha, n}^b$ from \eqref{eq:combRFBndry}, except that
we substitute the summation over the truncated lattice
$\Z_n^d$ with a series over the entire lattice $\Z^d$.
As the terms containing the sine-expansion don't involve
the constant function, we further introduce the index 
sets $\mc{I}_b$, replacing the $\mc{I}_{b, n}$ from
\eqref{eq:truncatedIndexSets} via
$$\mc{I}_b \coloneqq \mc{I}_{b_1} \times \ldots \times \mc{I}_{b_d},
        \quad \text{where} \quad \mc{I}_{0} \coloneqq \N \cup \{0\}
            \quad \text{and} \quad \mc{I}_{1} \coloneqq \N. $$
With this notation at hand, we can state the definition of $u_{\alpha}^b$
more precisely as
\begin{equation}    \label{eq:fullCombinedRF}
   u_{\alpha}^{b} (x, \omega) \coloneqq \suml{\mu \in \mc{I}_{b}} \xi_\mu (\omega)
            \sqrt{\farg{\varphi}{(2\alpha)^{-1} \mu}}~e_{\mu, \alpha}^{b}(x),
                \quad b \in \{0, 1\}^d, 
\end{equation}
with the stochastic coefficients $\xi_{\mu}$ as before and 
basis functions $e_{\mu, \alpha}^b$ from \eqref{eq:perRFModes}.

\begin{proposition} \label{prop:SPDECorrespondence}
Let $\varphi$ be a stationary \Matern covariance function with parameters
$\nu  \geq \nicefrac{1}{2}$ and $\ell > 0$, and consider the case $d \in \{1, 2, 3\}$. Then, for any
$b \in \{0, 1\}^d$, the GRF $u_{\alpha}^b$ as in \eqref{eq:fullCombinedRF} is a solution to 
the stochastic BVP~\eqref{eq:spdeBVP}.
\end{proposition}
\begin{proof}
    Recall the definition $L = \kappa^2 - \Delta$ from Section~\ref{sec:prel:spdeSampling}.
    We consider the weak formulation of the SPDE~\eqref{eq:SPDEProblem} given by
    \begin{equation}    \label{eq:spdeBVPWeak}
        \iprod{L^{\beta} u, v} = \sqrt{C_\nu} \kappa^{-\nu} \mc{W}(v), \quad \text{for all } v \in \ltwo{D_{\alpha}}
            \quad \text{a.s.},
    \end{equation}
    similar to \cite{bolin2020numerical, cox2020regularity, lindgren2011explicit}.
    The assumptions on the parameters $\nu$, $\ell$ and $d$ along with our choice
    of domain $D_{\alpha}$, ensure existence of a solution, as well as sufficient
    regularity of the domain of $L^{\beta}$, in the sense that $L^{\beta}$ can be
    extended to a bounded operator $L^{\beta}: H^1(D_{\alpha}) \to \ltwo{D_{\alpha}}$
    according to Lemma~2.1 in \cite{bolin2020numerical},
    in combination with Lemma~2(ii) and Remark~4 in~\cite{cox2020regularity}.
   
    The action of $L^{\beta}$ on its domain in \eqref{eq:fracPowerAction} is
    determined by the spectral properties of the underlying operator $L$.
    We observe that for any $b \in \{0, 1\}^d$, the
    $\{e_{\mu, \alpha}^{b}\}_{\mu \in \mc{I}_{b}}$
    are eigenfunctions of $L$. This can be seen from the fact that for any 
    $x \in \R^d$, 
    \begin{align}
        \big[L\,e_{\mu, \alpha}^{b}\big](x)
            &= \kappa^2 e_{\mu, \alpha}^{b} (x) + \Big(\sumlim{k=1}{d} \big(\pi \alpha^{-1} \mu_k\big)^2 \Big)
                e_{\mu, \alpha}^{b}(x) \\
            &= \left(\kappa^2 + \pi^2 \alpha^{-2} \norm{\mu}_2^2\right)~e_{\mu, \alpha}^{b}(x) \\
            &\eqqcolon \eta_{\mu, \alpha} e_{\mu, \alpha}^{b}(x).
    \end{align}
    Moreover, the eigenvalues $\eta_{\mu, \alpha}$ are closely related to the Fourier transform
    \eqref{eq:maternHankel} of the stationary \Matern covariance function $\varphi$ via
    \begin{equation} \label{eq:LbetaEVRepresentation}
        \eta_{\mu, \alpha}^{-2 \beta}
            = ( \kappa^2 + \pi^2 \alpha^{-2} \norm{\mu}_2^2)^{-2 \beta}
            = C_\nu^{-1} \kappa^{2 \nu}~\fourier{\varphi}((2 \alpha)^{-1} \mu),
    \end{equation}
    where $C_{\nu}$ is defined in \eqref{eq:maternDefinition}.
    In order to incorporate the
    boundary conditions \eqref{eq:spdeBVP_bndry}, we project onto the subspace spanned by those eigenfunctions
    corresponding to the fixed $b \in \{0, 1\}^d$ of the BVP \eqref{eq:spdeBVP}
    (c.f.~Section~2.2 in~\cite{khristenko2019analysis} and Example~2.2 in~\cite{bolin2020numerical}).
    As a consequence, Proposition~\ref{prop:SPDECorrespondence} reduces to the claim that $u_{\alpha}^b$ satisfies
    \begin{equation}    \label{eq:projectedSPDE}
        \iprod{L^{\beta} u_{\alpha}^b(\cdot, \omega), e_{\mu, \alpha}^b}
            = \sqrt{C_{\nu}} \kappa^{-\nu} \mc{W}(e_{\mu, \alpha}^b)(\omega)
                \quad \text{for all } \mu \in \mc{I}_b \quad \text{a.s.}
    \end{equation}
    Inserting $u_{\alpha}^b$ into the left-hand side of
    \eqref{eq:projectedSPDE} for a fixed $\omega \in \Omega$ and using orthonormality of the $e_{\mu, \alpha}^b$,
    as well as \eqref{eq:LbetaEVRepresentation} yields
    \begin{align} 
        \iprod{L^{\beta} u_{\alpha}^b(\cdot, \omega), e_{\mu, \alpha}^b}
            &= \suml{\gamma \in \mc{I}_b} \eta_{\gamma, \alpha}^{\beta} \iprod{u_{\alpha}^b(\cdot, \omega),
                e_{\gamma, \alpha}^b} \iprod{e_{\gamma, \alpha}^b, e_{\mu, \alpha}^b} 
           = \eta_{\mu, \alpha}^{\beta} \iprod{u_{\alpha}^b(\cdot, \omega), e_{\mu, \alpha}^b}.
            \label{eq:innerProdLbeta}
    \end{align}
    Further, the inner product occurring in above equation is given by
    \begin{equation}    \label{eq:innerProdU}
        \iprod{u_{\alpha}^b(\cdot, \omega), e_{\mu, \alpha}^b}
            = \suml{\gamma \in \mc{I}_b} \xi_{\gamma}(\omega)
                \big(\fourier{\varphi}((2\alpha)^{-1})\big)^{\nicefrac{1}{2}}
                    \iprod{e_{\gamma, \alpha}^b, e_{\mu, \alpha}^b}
            =  \xi_{\mu} (\omega) \big(\fourier{\varphi}((2\alpha)^{-1})\big)^{\nicefrac{1}{2}}.
    \end{equation}
    Combining \eqref{eq:innerProdLbeta} and \eqref{eq:innerProdU}, equation \eqref{eq:projectedSPDE}
    becomes
    \begin{equation}    \label{eq:whiteNoiseCoeff}
        \xi_{\mu}(\omega) = \mc{W}(e_{\mu, \alpha}^b)(\omega), \quad \text{for all } \mu \in \mc{I}_b,
        ~a.s.
    \end{equation}
    Finally, as the $e_{\mu, \alpha}^b$ are orthonormal, the Gaussian random variables 
    $\xi_{\mu}^{\mc{W}} \coloneqq \mc{W}(e_{\mu, \alpha}^b)$ are uncorrelated
    and each follow a standard normal distribution according to \eqref{eq:whiteNoiseDef},
    which implies \eqref{eq:whiteNoiseCoeff} and thus concludes the proof.
\end{proof}

Due to
Proposition~\ref{prop:SPDECorrespondence}, the framework developed
in Section~\ref{sec:higherDim} can be readily applied to solutions to the SPDE \eqref{eq:spdeBVP},
as demonstrated in the following Corollary.
\begin{corollary}
    Let $\varphi$ be a stationary \Matern covariance function with parameters
    $\nu \geq \nicefrac{1}{2}$, $\ell > 0$ and $d \in \{1, 2, 3\}$. Define
    the DNA GRF $u_{\alpha}$ by
    \begin{equation}    \label{eq:fullCombSPDE}
        u_\alpha(x, \omega) \coloneqq 2^{-\nicefrac{d}{2}} \! \suml{b \in \{0, 1\}^d} u_{\alpha}^b(x, \omega),
            \quad x \in D_\alpha,~\omega \in \Omega,
    \end{equation}
    where for each $b \in \{0, 1\}^b$ the $u_{\alpha}^b$ are solutions to the corresponding BVP \eqref{eq:spdeBVP}.
    Then, $u_{\alpha}$ is isotropic and its covariance function is given by $\prdi{\varphi}{2 \alpha}$. 
    If in addition $2 \alpha \kappa > \nicefrac{3}{2}$ holds, then
    $\prdi{\varphi}{2 \alpha}$ satisfies the error bound \eqref{eq:covErrorMatern} from
    Lemma~\ref{lem:periodError}.
\end{corollary}

A notable advantage of the DNA approach is that the covariance of the DNA GRF \eqref{eq:fullCombSPDE}
is genuinely isotropic for any value of $\alpha$. The covariance function is given by $\prdi{\varphi}{2 \alpha}$.
In contrast, using a single set of boundary conditions on an oversampling domain $\widetilde{D} \supset D$ the
GRFs produced by the SPDE approach are never truly isotropic, even for large amounts of oversampling. It can only
be shown that the error decays roughly at the same rate as the \Matern covariance function with the oversampling
size \cite[Theorem 3.2]{khristenko2019analysis}.
Moreover, under very mild assumptions, Lemma~\ref{lem:periodError} provides an explicit upper
bound on the covariance error for the resulting GRF, irrespective of the method used for the
subsequent discretisation of the BVP~\eqref{eq:spdeBVP}.

\begin{remark}
    A key consideration for sampling based on the SPDE approach is that minor discrepancies
    in the covariance function of the continuous problem are inconsequential, as they are dominated by the
    inevitable errors due to discretising the SPDE. These discretisation errors can be shown (both
    theoretically and in practice) to decay at most algebraically, with a rate limited by the
    smoothness of the GRF \cite{simpson2016going, cox2020regularity},
    in contrast to the exponential decay with respect to the parameter $\alpha$ of the DNA 
    periodisation error in Lemma~\ref{lem:periodError}.
    This aligns well with the other advantages of DNA sampling, which offers isotropic GRFs and natural
    parallelisation of the simulation.
    We will explore this aspect further in the numerical
    experiments below.
\end{remark}

\setlength{\belowcaptionskip}{-1em}
\section{Numerical Experiments} \label{sec:numExp}
    The first central result in Section~\ref{sec:sampFr} is the construction of
the periodisation-based sampling framework. Figure~\ref{fig:periodisation} illustrates the 
three kinds of periodisation introduced in Section~\ref{sec:sampFr}: (i) \;the naive periodisation
$\prdi{\varphi}{\alpha}$ from~\eqref{def:periodisation}, attained by evaluating~\eqref{def:periodisedRF}
and discarding half the computed realisation; (ii) \;the periodisation $\prdi{\varphi}{2 \alpha}$ for the DNA GRF
\eqref{def:combRF} as per Proposition~\ref{prop:combRFCov}; and (iii) \;the periodisation
$\prdi{\varphi}{\alpha, \mrm{CE}}$ as in \eqref{eq:circulantEmbeddingPeriodisation},
implied by the sampling procedure of CE.

As discussed in Section~\ref{sec:prel:prd}, sampling with the exact covariance on the target
domain using CE is possible only if the corresponding periodisation, $\prdi{\varphi}{\alpha, \mrm{CE}}$,
is positive-definite, in the sense of Definition~\ref{def:periodisation}. For Matérn covariances, this
requirement can 
be met through expanding the computational domain—referred to as \enquote{padding}
\cite{graham2018analysis}. Within the periodisation framework of
Section~\ref{sec:prel:prd}, this expansion corresponds to increasing $\alpha$.

In the one-dimensional case, Figure~\ref{fig:overhead} illustrates the factor $\tau$ by which
the FFT input length in CE exceeds that for a single one of the $2^d = 2$ DSTs/DCTs used in 
DNA sampling with the same parameters. 
Although obtaining two realisations from the real and imaginary parts of the CE evaluation makes the
effort comparable for fields with low regularity (characterised by small $\nu$ and
short correlation lengths $\ell$), the effort remains sensitive to these parameters and increases
sharply as either $\nu$ or $\ell$ grows. Values of $\nu$ higher than those
in Figure~\ref{fig:overhead} did not provide a positive-definite periodisation with a factor of less than $\tau = 256$
for the given correlation lengths.
\begin{figure}[ht] 
    \centering
    \begin{minipage}[t]{0.48\linewidth}
        \centering
        \includegraphics[width=\linewidth]
    {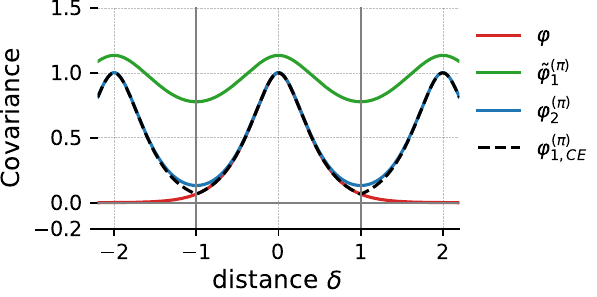}
        \caption{Three periodisations of the \Matern covariance function $\varphi$ with parameters $\nu = 2$ and
        $\ell = \nicefrac{2}{5}$, corresponding to the sampling algorithms discussed in Section~\ref{sec:sampFr}.
        The scaling parameter is $\alpha = 1$.}
        \label{fig:periodisation}
    \end{minipage} \hfill
    \begin{minipage}[t]{0.48\linewidth}
        \centering
        \includegraphics[width=\linewidth]
        {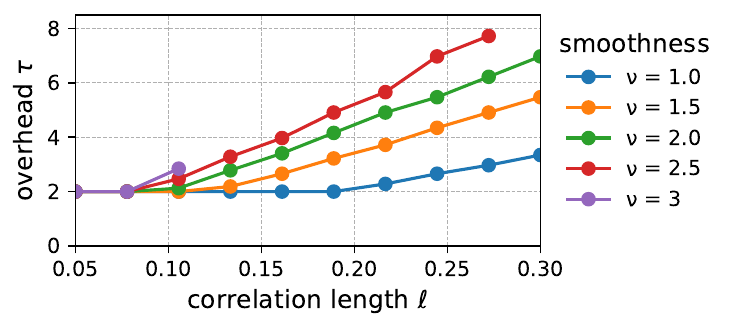}
        \caption{
        Overhead for \Matern covariance with smoothness $\nu$ on $D=(0,1)$ with
                 $n=1500$ grid points. Minimal embedding determined via bisection; missing values indicate
                 padding exceeding 256.} 
        \label{fig:overhead}
    \end{minipage}\vspace{3ex}
\end{figure}

A key idea in the DNA sampling framework is that modifying
the basis functions of a periodised field \eqref{def:periodisedRF} corresponds to altering the
boundary conditions prescribed to realisations of the associated GRF. Figure~\ref{fig:covBndryConditions} presents
heat-maps of empirical covariance matrices for periodic and homogeneous Dirichlet/Neumann boundary conditions.
For the one-dimensional case, the core observation in Section~\ref{sec:sampFr} is that the analytically
derived covariance functions for homogeneous Dirichlet and homogeneous Neumann boundary conditions
in \eqref{eq:dirCov} and \eqref{eq:neuCov} contain anisotropic terms that precisely cancel each other
out. This effect is illustrated in Figure~\ref{fig:dirneuCov}, which shows the empirical marginal variances
for the 1D fields.
\renewcommand{\topfraction}{.8}
\renewcommand{\floatpagefraction}{.7}
\begin{figure}[ht] 
    \centering
    \begin{minipage}[t]{0.48\linewidth}
        \centering
        \includegraphics[width=0.8\linewidth]{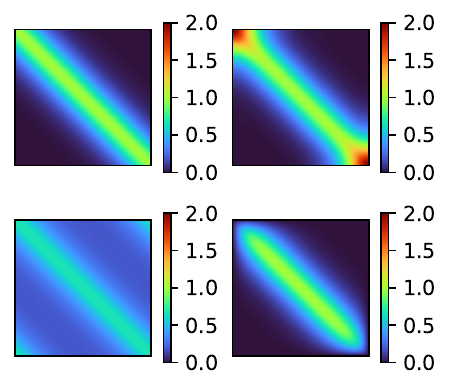}
        \caption{Heatmaps of target (top left) and empirical covariances
        for \Matern fields ($\nu = 2$, $\ell = 0.15$) with 
        $N = 5 \cdot 10^5$ realisations. Clockwise from top right:
        homogeneous Neumann, homogeneous Dirichlet, periodic boundary conditions.}
        \label{fig:covBndryConditions}
    \end{minipage}
    \hfill
    \begin{minipage}[t]{0.48\linewidth}
        \centering
        \includegraphics[width=\linewidth, height=0.65\linewidth]{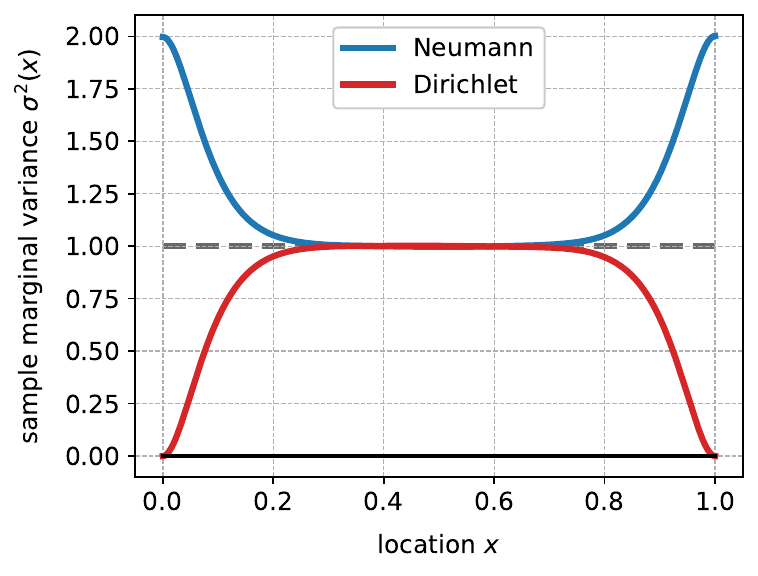}
        \caption{Empirical marginal variances for a GRF with \Matern covariance function
        with $\nu = 2.0$ and $\ell = 0.15$, estimated using $N = 5 \cdot 10^5$ realisations.
        Values correspond to diagonals of the empirical covariance matrices in
        Figure~\ref{fig:covBndryConditions}.}
        \label{fig:dirneuCov}
    \end{minipage}
\end{figure}

In two dimensions, generating a realisation of the DNA GRF requires averaging $2^d = 4$
realisations, each imposed with a different boundary condition. As illustrated in Figure~\ref{fig:margVar2d}, this
construction yields a constant marginal variance across the domain, aligning with the expected behaviour of
isotropic fields. The figure also shows how the four GRFs with different boundary
conditions, which evidently have non-uniform marginal variances, interact to
produce a constant marginal variance.

\begin{figure}[ht] 
    \centering
    \includegraphics[width=0.6\linewidth, height=0.34\linewidth]{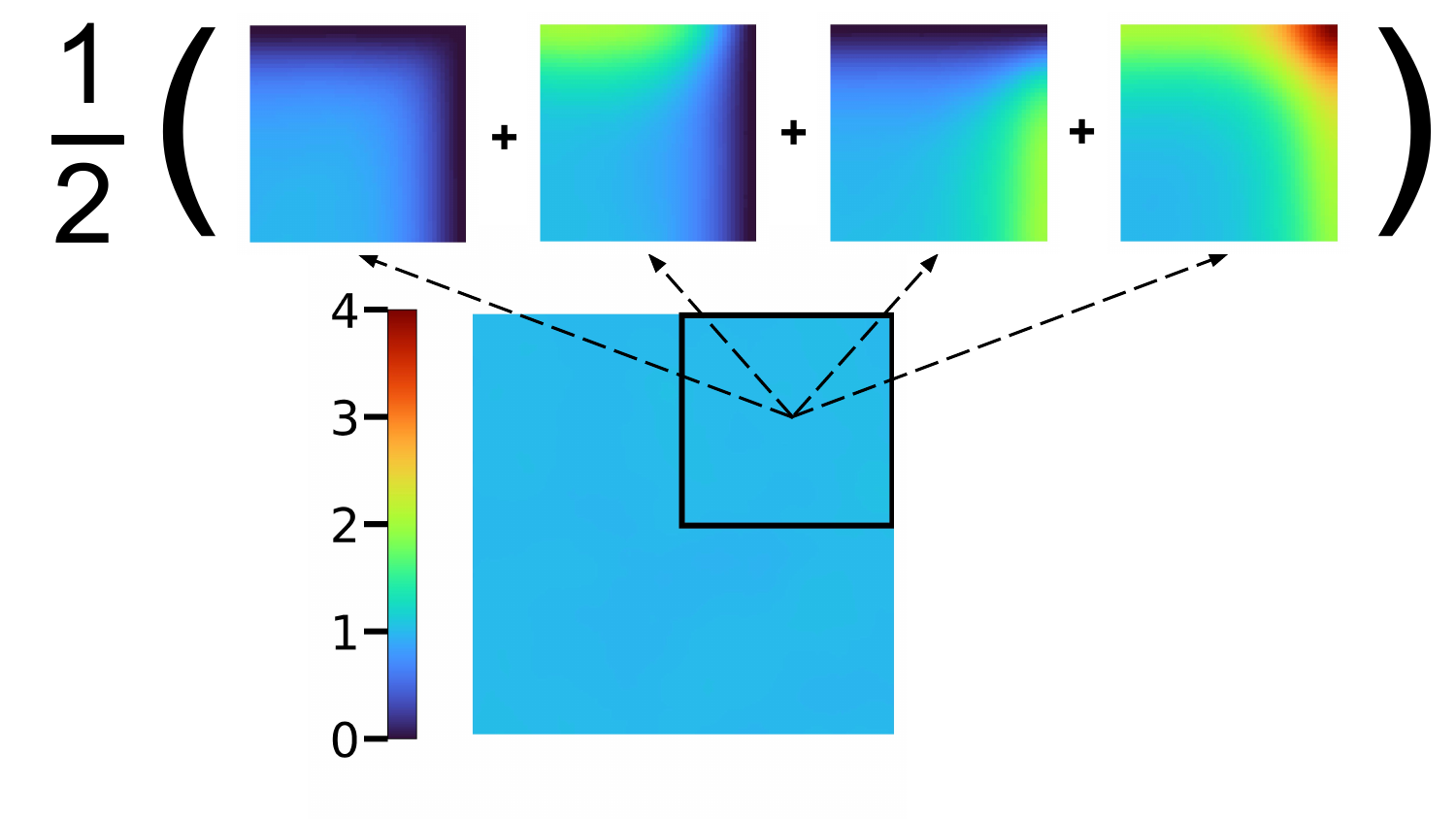}
    \caption{Empirical marginal variances for the $2^d = 4$ combinations of boundary conditions
    used for DNA sampling, along with the resulting empirical marginal variance
    of the DNA GRF. Parameters for the \Matern covariance: $\nu = 1.5$, $\ell = 0.2$; estimated from
    $N = 10^4$ realisations on a $150 \times 150$ grid.}
    \label{fig:margVar2d}
\end{figure}

In several places throughout the paper, it was noted that the errors due to periodisation in the
DNA approach are in general negligible in practice.
Table~\ref{tab:covErrors} provides estimates of the maximal covariance error, as defined in
\eqref{eq:covErrorDef}, for five different covariance functions in the case $d=1$. Among these
functions are the Gaussian covariance function, $\varphi_{\mrm{G}}$, and the Cauchy covariance
function, $\varphi_{\mrm{C}}$, which for any $\delta \in \R$ are defined through
$\varphi_{\mrm{G}} (\delta) = e^{-\frac{\delta^2}{2 \ell^{\,2}}}$ and
$\varphi_{\mrm{C}}(\delta) = (1 + \ell^{-2} \delta^2)^{-1}$.
The required Fourier transforms are given by
$\fourier{\varphi}_{\mrm{G}} (s) = \sqrt{2 \pi} \ell e^{-2 \pi^2 \ell^{\,2} s^2}$
and $\fourier{\varphi}_{\mrm{C}} (s) = \pi \ell e^{-2 \pi \ell |s|}$
respectively.
\begin{table}[ht] 
\centering
\begin{tabular}{l|ccccc}
    $\ell$  & $\nu = 0.5$ & $\nu = 2$ & $\nu = 8$ & Gaussian & Cauchy \\
    \midrule
$0.025$ & 1.77e-2 & 1.33e-2 & 1.30e-2 & 1.24e-2 & 1.30e-2 \\
$0.05$  & 1.53e-2 & 1.16e-2 & 1.13e-2 & 1.11e-2 & 1.36e-2 \\
$0.1$   & 1.39e-2 & 1.08e-2 & 9.3e-3 & 9.8e-3 & 1.83e-2 \\
$0.2$   & 1.31e-2 & 8.3e-3 & 8.9e-3 & 8.3e-3 & 5.63e-2 \\ 
\end{tabular}
\caption{Monte-Carlo estimates of the maximal covariance error as defined in \eqref{eq:covErrorDef}
based on $N = 1.6 \cdot 10^6$ realisations on a grid with $n = 1500$ points.}
\label{tab:covErrors}
\end{table}
The truncation parameter $n$ in \eqref{def:combRF} controls the highest frequencies
occurring in a periodised random field. Consequently, as Lemma~\ref{lem:truncError} suggests, for a
fixed $n$, the truncation error tends to dominate for fields with low regularity, e.g.~small correlation
lengths $\ell$ or low values of $\nu$ in the \Matern case. The periodisation error, on the other hand,
behaves in the opposite way: as Lemma~\ref{lem:periodError} indicates, high regularity, for example manifested
in higher values of $\nu$ and particularly $\ell$ lead to a greater periodisation error, which then
becomes the primary error contribution. This balance is visible in Table~\ref{tab:covErrors}, reflecting
how truncation and periodisation errors trade off under different parameter values.
As the error in Table~1 decreases with increasing correlation length $\ell$ for the \Matern and Gaussian
covariances, one can infer that the discretisation error still dominates in these cases. Only for the highly
regular realisations associated with the Cauchy covariance, the error increases with $\ell$, indicating that
here, the periodisation error dominates.
The overall error in Table~\ref{tab:covErrors} is in the low single-digit percentages, which
seems acceptable for most applications.

The GRFs for the covariances presented in Table~\ref{tab:covErrors} were generated with $\alpha = 1$,
thus without requiring any padding or domain extension. Even for slowly decaying covariance functions
like the Cauchy covariance, this setup yielded relatively low error. In contrast, for the CE method, it
was not possible to achieve a positive-definite embedding for the $\nu = 8$, Gaussian, and Cauchy
covariance functions, despite allowing for a domain extension of up to 1024 times the original size,
with the same parameter settings. This highlights a substantial advantage of the DNA
sampling approach in handling such covariances efficiently.

In Figure~\ref{fig:spdeRates}, we compare the maximal covariance error along a diagonal of the target
domain, $D = [0, 1]^2$, for the SPDE sampling approach with domain extension and the DNA sampling approach.
This comparison is performed for various values of $\alpha$,
which correspond to different domain extension sizes. As expected, increasing the size of the domain
extension reduces the error for both approaches. However, the absolute error magnitude is significantly
lower with the averaging approach.
\begin{figure}[ht]
    \centering
    \includegraphics[width=0.85\linewidth, height=0.43\linewidth]{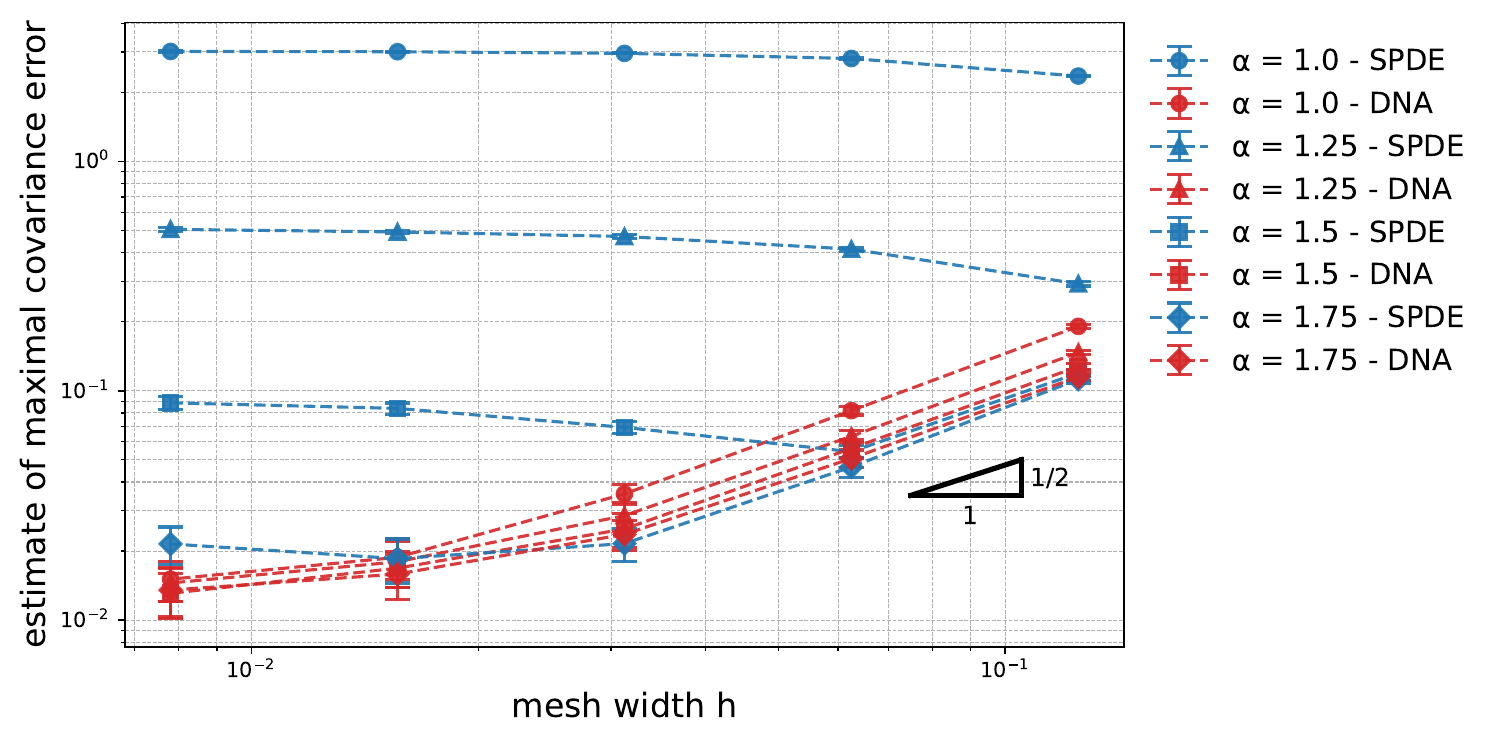}
    \caption{Monte-Carlo estimates of the maximal covariance error along the diagonal of the target
    domain $D=(0,1)^2$, based on $N=2 \cdot 10^6$ realisations. Parameters for the SPDE
    are $\nu = 1.0$ and $\ell = 0.25$. We compare the DNA sampling approach from
    Section~\ref{sec:sampFr} with the classical SPDE approach with homogeneous Neumann boundary
    conditions on the entire boundary under varying domain extensions. Error bars correspond
    to one empirical standard deviation of the error, estimated using cross-validation with
    $40$ batches of $5 \cdot 10^4$ samples each.}
    \label{fig:spdeRates}
\end{figure}

In the SPDE approach, the total covariance error naturally separates into two components: one from the
FE discretisation and another from the intrinsic covariance error. Figure~\ref{fig:spdeRates} suggests that, when
comparing the two methods, the averaging approach allows the discretisation error to dominate across a
wider range of mesh widths. Consequently, refining the mesh in the averaging approach leads to further
error reduction and thus higher accuracy.
For the classical approach with domain extension, a common heuristic recommends extending the domain by
roughly one to two correlation lengths (see, e.g., \cite{khristenko2019analysis, lindgren2011explicit, roininen2014whittle}).
In Figure~\ref{fig:spdeRates}, we see that for the considered example, with $\ell = \nicefrac{1}{4}$
and $\nu = 1$, following this heuristic would mean that refining the mesh to a width smaller than
$h = \nicefrac{1}{16}$ or $h = \nicefrac{1}{32}$ respectively, would no longer decrease the error,
as the covariance error then dominates.

The reduced error in the DNA sampling approach comes with an increased computational cost,
though this overhead is in general manageable. Excluding the (negligible) averaging and setup costs
for boundary conditions, a single evaluation of the DNA GRF in two dimensions is four
times as expensive as using the standard SPDE approach with the same domain extension.

Although the error in the DNA GRF decreases with domain extension, the improvements are less
pronounced than in the oversampling 
approach and they are dominated by the discretisation error across most
tested mesh widths. This suggests that, for many practical purposes, it may be advantageous to avoid
increasing the domain size altogether. Doing so simplifies implementation and lowers computational costs
for comparable error.

While the extension of the DNA approach to SPDE-based sampling in Proposition~\ref{prop:SPDECorrespondence}
assumes that the domain $D$ is a hypercube
(c.f.~Remark~\ref{rem:boundingBox}), it is interesting to explore whether locally
averaging GRFs with Dirichlet and Neumann boundary conditions extends to more general domains.
Figure~\ref{fig:nonsquare} demonstrates this idea for two non-cuboid domains. For the unit disc
in Figure~\ref{subfig:nonsquare:disc}, which has no corners, only two GRFs with homogeneous
Neumann and Dirichlet boundary conditions were averaged, instead of the usual $2^d = 4$ GRFs. The resulting
averaged GRF shows constant marginal variance, indicating isotropic covariance, in contrast to the
clearly anisotropic covariances for homogeneous Neumann/Dirichlet conditions.
A similar effect appears in the convex regions of the more general domain in
Figure~\ref{subfig:nonsquare:ship}, but at least the basic DNA approach seems to break down
in non-convex corners.
\begin{figure}[bt] 
    \centering
    \begin{subfigure}{0.48\textwidth}
        \centering
        \includegraphics[width=0.8\linewidth, height=0.6\linewidth]{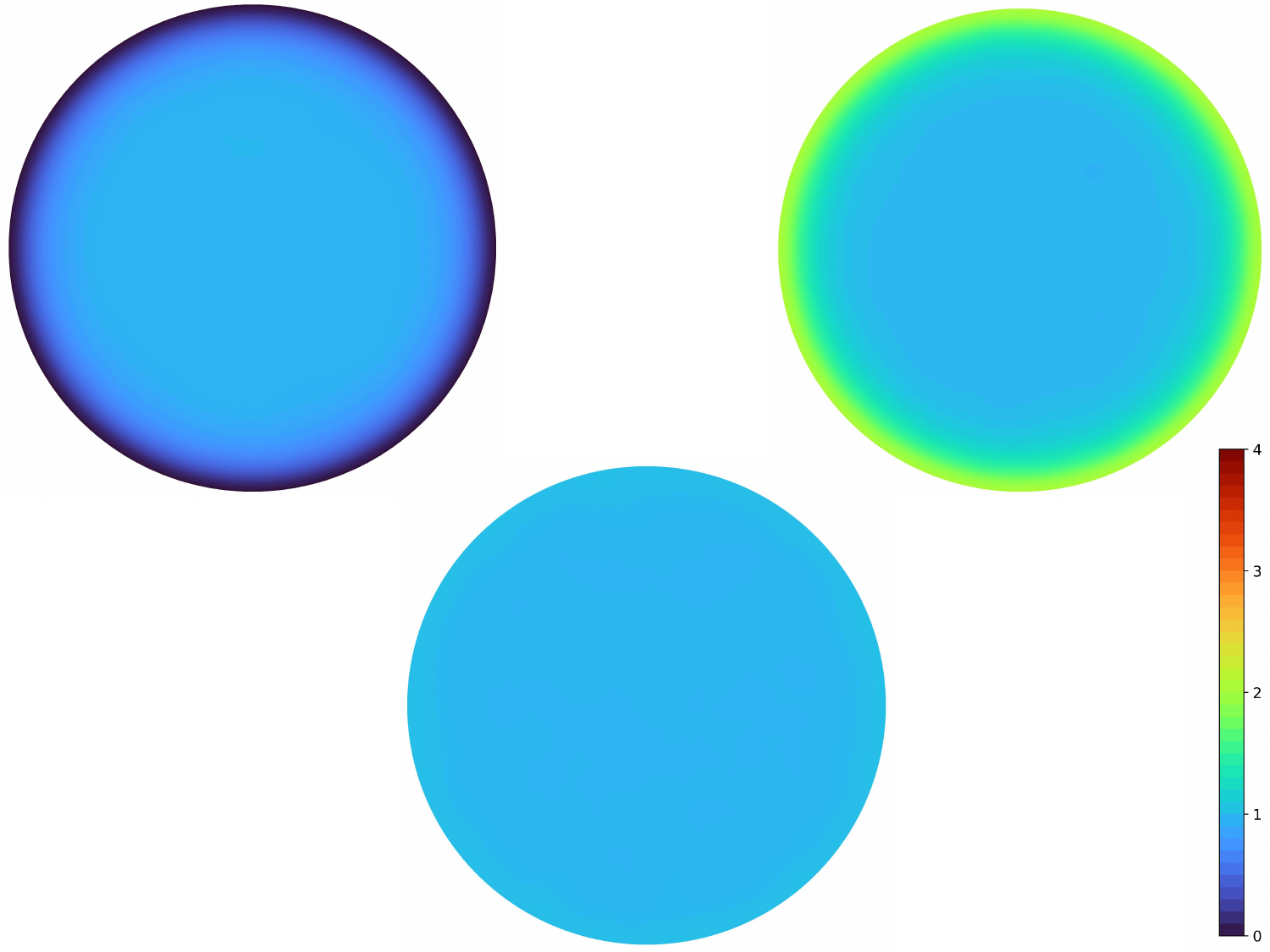} 
        \caption{Unit disc with $n = 5400$ cells.}
        \label{subfig:nonsquare:disc}
    \end{subfigure}
    \begin{subfigure}{0.48\textwidth}
        \centering
        \includegraphics[width=0.8\linewidth, height=0.6\linewidth]{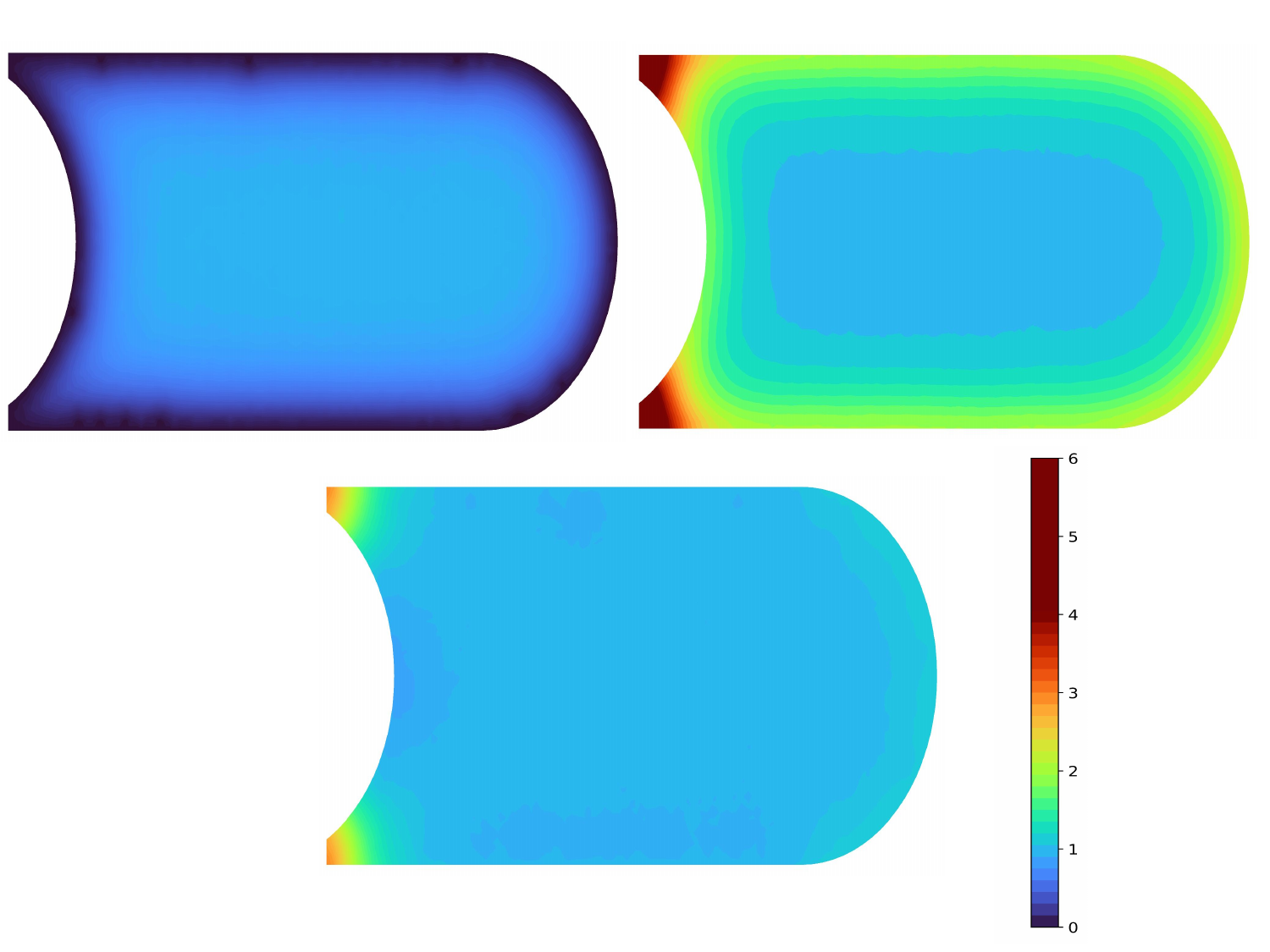} 
        \caption{Irregular domain with $n=6563$ cells.}
        \label{subfig:nonsquare:ship}
    \end{subfigure}\vspace{2ex}
    \caption{Marginal variance of GRFs with \Matern covariance ($\nu = 1$, $\ell = 0.2$) for
    different boundary conditions (BCs), computed via the SPDE
    approach ($10^5$ realizations). Top rows: homogeneous Dirichlet and Neumann BCs; bottom: averaged GRFs.}
    \label{fig:nonsquare}
\end{figure}

Our experiments indicate that averaging homogeneous Dirichlet and Neumann boundary conditions can yield
isotropic covariance in cases where symmetry or local convexity may play a role. This observation has
potential implications for GRF generation on manifolds, such as spheres, where boundary effects remain an open problem.
Notably, using two boundary conditions reduces computational cost compared to the $2^d$ required for the
averaged GRF.

\printbibliography

\appendix 
\section{Auxiliary Results} \label{appendix} 
    \begin{lemma}   \label{app:lem:maternAssumption}
    Let $\varphi$ be a stationary covariance function from the \Matern family
    with correlation length $\ell > 0$ and smoothness parameter $\nu > 0$, as
    defined in \eqref{eq:maternDefinition}. Then $\varphi$ satisfies
    Assumption~\ref{assump:statCov} with
    \begin{equation}
        A = C_\nu \ell^{\,d} \nu^{-\nicefrac{d}{2}} 2^{2 \nu + \nicefrac{d}{2}}
            \quad \text{and} \quad
            \varepsilon = 2 \nu.
    \end{equation}
\end{lemma}
\begin{proof}
    The integrability assumption in~\ref{assump:statCov} is satisfied as
    per the boundedness and exponential decay of $\varphi$, proven e.g.~in
    \cite{bachmayr2020unified},~Lemma~2.
    In proving the decay property, we note that
    according to \eqref{eq:maternHankel} for any $\delta \in \R^d$,
    the Fourier transform $\fourier{\varphi}$ of the stationary \Matern covariance
    function can be rewritten as
    \begin{equation}
        \fourier{\varphi}(\delta) =
            C \left(1  + 4 \pi^2 \kappa^{-2} \norm{\delta}_2^2\right)^{-2 \beta},
                \quad \text{where} \quad C \coloneqq \kappa^{-d} C_{\nu},
    \end{equation}
    and $\kappa, \beta > 0$ are as defined in Section~\ref{sec:prel:spdeSampling}.
    We claim that $h(\delta) \leq A$ for all $\delta \in \R^d$, where we define $h$ as the ratio\vspace{-2ex}
    \begin{equation}
        \qquad h(\delta) \coloneqq \frac{\fourier{\varphi}(\delta)}{(1 + \norm{\delta}_2)^{-(d + 2 \nu)}}.
    \end{equation}
    For $\delta = 0$, we have $h(0) = C \leq A$. Next, in the case that
    $0 < \norm{\delta}_2 \leq 1$, we have
    $h(\delta) \leq C 2^{d + 2 \nu} = A$. Finally, if $\norm{\delta}_2 > 1$,
    then
    \begin{equation}
        h(\delta) \leq \frac{C 2^{d + 2 \nu}}
            {(\norm{\delta}_2^{-2} + 4 \pi^2 \kappa^{-2})^{2 \beta}}
            \leq \frac{C 2^{d + 2 \nu}}{(1 + 4 \pi^2 \kappa^{-2})^{2 \beta}}
            \leq C 2^{d + 2 \nu} = A,
    \end{equation}
    which concludes the proof.
\end{proof}

\begin{lemma}   \label{app:lem:auxSum}
Let $\beta, \gamma \in \R$ with $0 \leq \beta < \gamma$ and
$0 < \vartheta \coloneqq \Gamma(2+\beta)^{-\nicefrac{1}{(1+\beta)}} < 1$. Then
\begin{equation}
        \sumlim{m = 1}{\infty} m^{\beta} e^{-\gamma m}
            \leq e^{-\vartheta \gamma} \Big(1 + \Gamma(2 + \beta)~\gamma^{-(1+\beta)}\Big).
\end{equation}
\end{lemma}
\begin{proof}
    The sole maximum of the positive and continuous function
    $g: [0, \infty) \to \R$, with $x \mapsto x^{\beta} e^{-\gamma x}$,
    is attained at $x_{\max} = \nicefrac{\beta}{\gamma}$. The assumption
    $\beta < \gamma$ then assures that $x_{\max} < 1$ and consequently, $g$
    is monotonously decreasing on $[1, \infty)$. In this case, a simple 
    integral test estimate gives
    \begin{equation}    \label{eq:intBound0}
        \sumlim{m=1}{\infty} m^{\beta} e^{-\gamma m}
            = e^{-\gamma} + \sumlim{m=2}{\infty} m^{\beta} e^{-\gamma m}
            \leq e^{-\gamma} + \int\limits_{1}^{\infty} t^{\beta} e^{-\gamma t}~\mrm{d}t.
    \end{equation}
    Continuing with the second term, we see that
    \begin{equation}    \label{eq:intGamma}
        \int\limits_{1}^{\infty}t^{\beta} e^{-\gamma t}~\mrm{d}t
            = \gamma^{-(1+\beta)} \int\limits_{\gamma}^{\infty}t^{\beta} e^{-t}~\mrm{d}t
            \equiv \gamma^{-(1+\beta)}~\Gamma (1 + \beta, \gamma),
    \end{equation}
    where $\Gamma$ denotes the upper incomplete Gamma function. The case $\beta = 0$ is 
    straight-forward, so we assume $\beta > 0$ in the following. Now, making the change of 
    variables $t \mapsto t^p$, where $p \coloneqq (1+\beta)^{-1}$, we see that
    $\Gamma(1+\beta, \gamma) = p \int_{\gamma^{\nicefrac{1}{p}}}^{\infty}e^{-t^p} \mrm{d}t$,
    and therefore Theorem~1 and its Corollary from \cite{alzer1997some} imply that in the present case,
    \begin{equation}    \label{eq:intermediateGamma}
        \Gamma(1+\beta, \gamma)
            < p \, \Gamma\big(1 + \nicefrac{1}{p}\big) 
                \left[1 - \Big(1 - e^{-\vartheta \gamma} \Big)^{\nicefrac{1}{p}}\right],
    \end{equation}
    for any $0 \leq \vartheta \leq \min\big\{1, \Gamma\big(1 + \nicefrac{1}{p}\big)^{-p}\big\}$.
    Inserting \eqref{eq:intermediateGamma} into \eqref{eq:intGamma}, as well as using the fact that
    $(1 - x)^{y} \geq 1 - xy$ for $x \in [0, 1]$ and $y \geq 1$, yields
    \begin{equation}    \label{eq:intBound1}
        \int\limits_{1}^{\infty}t^{\beta} e^{-\gamma t}~\mrm{d}t \leq \gamma^{-(1 + \beta)}~\Gamma(2 + \beta) e^{-\vartheta \gamma}.     
    \end{equation}
    We note that for $\beta \geq 0$, the largest value for $\vartheta$ (i.e.~quickest decay) we can choose is
    $\vartheta = \Gamma(2 + \beta)^{-\nicefrac{1}{(1+\beta)}}$. As $\vartheta \leq 1$, we have
    $e^{-\gamma} \leq e^{-\vartheta \gamma}$ and applying this bound, as well as the bound \eqref{eq:intBound1} 
    to \eqref{eq:intBound0} yields the desired statement.
\end{proof}

\end{document}